\def\year{2022}\relax
\documentclass[letterpaper]{article}
\usepackage{aaai22} % DO NOT CHANGE THIS
\usepackage{times} % DO NOT CHANGE THIS
\usepackage{helvet} % DO NOT CHANGE THIS
\usepackage{courier} % DO NOT CHANGE THIS
\usepackage[hyphens]{url} % DO NOT CHANGE THIS
\usepackage{graphicx} % DO NOT CHANGE THIS
\usepackage{graphicx} % DO NOT CHANGE THIS
\usepackage{natbib} % DO NOT CHANGE THIS
\usepackage{caption} % DO NOT CHANGE THIS
\DeclareCaptionStyle{ruled}%
{labelfont=normalfont,labelsep=colon,strut=off}
\usepackage{tikz}
\usetikzlibrary{arrows,chains,matrix,positioning,scopes}
\makeatletter
\tikzset{join/.code=\tikzset{after node path={%
\ifx\tikzchainprevious\pgfutil@empty\else(\tikzchainprevious)%
edge[every join]#1(\tikzchaincurrent)\fi}}}
\makeatother
\tikzset{>=stealth',every on chain/.append style={join},
         every join/.style={->}}
\tikzstyle{labeled}=[execute at begin node=$\scriptstyle,
   execute at end node=$]

\usepackage{centernot}
\usepackage{times}
\usepackage{soul}
\usepackage{url}
\usepackage[hidelinks,colorlinks=true,urlcolor=blue,linkcolor=black,citecolor=black,anchorcolor=black]{hyperref}
\usepackage[utf8]{inputenc}
\usepackage{graphicx}
\usepackage{amsmath}
\usepackage{comment}
\usepackage{amsthm}
\usepackage{amsfonts}
\usepackage{booktabs}
\usepackage{color}
\usepackage{algorithmic}
\usepackage[linesnumbered,ruled,vlined]{algorithm2e}
\newcommand\full{}
\newcommand\extraTPS{}

\newtheorem{example}{Example}
\newtheorem{theorem}{Theorem}
\newtheorem{definition}{Definition}

\newtheorem{remark}{Remark}
\newtheorem{prop}{Proposition}

\newtheorem{corollary}{Corollary}
\newtheorem{claim}{Claim}

\newcommand{\efxwc}{\text{EFX}_{\text{WC}}}

\newcommand{\efowc}{\text{EF1}_{\text{WC}}}
\newcommand{\eflwc}{\text{EFL}_{\text{WC}}}
\newcommand{\efx}{\text{EFX}}

\newcommand{\mms}{\text{MMS}}
\newcommand{\efl}{\text{EFL}}

\newcommand{\dual}[1]{{#1}^{\circ}}

\newcommand{\bundle}{B} %symbol for bundle
\newcommand{\bundlei}[1][i]{\bundle_{#1}} %for bundle i
\newcommand{\bundles}{\mathbf{\bundle}} %for bundle collection

\newcommand{\alloc}{A} %symbol for allocation
\newcommand{\allocs}{\mathbf \alloc} %for a whole allocation 
\newcommand{\alloci}[1][i]{\alloc_{#1}} %for agent i's allocation, e.g. \alloci for A_i, and \alloci[j] for A_j
\newcommand{\alloclist}[1][n]{\left(\alloc_1,\dots,\alloc_{#1}\right)} %for full writing of an allocation

\newcommand{\agents}{\mathcal{N}} %for set of agents
\newcommand{\items}{\mathcal{M}} %for set of items
 
\newcommand{\types}{\mathcal{T}}
\newcommand{\ea}[1]{\mathcal{X}\left(#1\right)}

\newcommand{\vau}[1][i]{v_{#1}}
\newcommand{\val}[1][i]{v_{#1}} %for agent i's valuation function e.g. \val for v_i and \val[j] for v_j
\newcommand{\vai}[2][i]{v_{#1}\left(#2\right)} %for agent i's valuation function on a set of items e.g. \vai{M} for v_i(M), and \vai[j]{M} for v_j(M)

\title{\fontsize{16}{20} \selectfont Unified Fair Allocation of Goods and Chores via Copies}
\author{Yotam Gafni$^1$, Xin Huang$^1$, Ron Lavi$^{1,2}$ \& Inbal Talgam-Cohen$^1$ \\ % All authors must be in the same font size and format. Use \Large and \textbf to achieve this result when breaking a line
$^1$ Technion - Israel Institute of Technology \\
$^2$ University of Bath, UK \\
\{yotam.gafni@campus., xinhuang@campus., ronlavi@ie., italgam@cs.\}technion.ac.il
}

\begin{document}

\maketitle
\thispagestyle{plain}
\pagestyle{plain}

%\pagenumbering{arabic}
%\pagenumbering{gobble}

\begin{abstract}
We consider fair allocation of indivisible items in a model with goods, chores, and copies, as a unified framework for studying: (1)~the existence of EFX and other solution concepts for goods with copies; (2)~the existence of EFX and other solution concepts for chores. 
We 
establish 
a tight relation between these issues via two conceptual contributions: First, a refinement of envy-based fairness notions that we term envy \emph{without commons} (denoted $\efxwc$ when applied to EFX). Second, a formal \emph{duality theorem} relating the existence of a host of (refined) fair allocation concepts for copies to their existence for chores.
We demonstrate the usefulness of our duality result by using it to characterize the existence of EFX for chores through the dual environment, as well as to prove EFX existence in the special case of leveled preferences over the chores. 
We further study the hierarchy among envy-freeness notions without commons and their $\alpha$-MMS guarantees, showing for example that any $\efxwc$ allocation guarantees at least $\frac{4}{11}$-MMS for goods with copies.
%For goods, envy-freeness notions can be sorted into a hierarchy (namely, $\text{EFX} \implies \text{EFL} \implies \text{EF1}$), with the notions separated by their $\alpha$-MMS guarantees, as shown by Amanatidis \emph{et al}.~\shortcite{comparingNotions2018}. We are able to establish a similar hierarchy for our more general goods with copies model. 
%This unified framework relies on a new solution concept that we term $\efxwc$ and that we believe may be of independent interest. 
%It is a strict relaxation of EFX that is appropriate for goods with copies, and 
%its existence implies the existence of EFX for chores. 
%We believe envy without commons and in particular $\efxwc$ may also be of independent interest -- for example, any $\efxwc$ allocation guarantees at least $\frac{4}{11}$-MMS for goods with copies. , whereas EF1 guarantees no better than $\frac{1}{n}$-MMS for goods and $n$ agents.
\end{abstract}

\section{Introduction}

%This work studies the relation between existence of various solution concepts for fair allocation of indivisible items, for two different models that at first blush do not seem tightly related. 
This work studies the relation between the existence of a fair allocation of indivisible items among agents -- for various notions of fairness -- in two different models that at first blush do not seem tightly related. 
The first model has indivisible goods with multiple {\em copies} of each good. The second model has indivisible {\em chores} (or ``bads''), which are items with negative values.

\noindent {\bf Goods with copies.} While the seminal work of Budish~\shortcite{budish2011combinatorial} on allocation of indivisible items includes copies as a key ingredient (motivated by important applications like the allocation of university courses to students), %or the assignment of tasks to workers), 
most subsequent works focus on a single copy of each item \cite[see, e.g.,][]{caragiannis2019unreasonable}. We consider a model where each agent may receive at most one copy of each good; we term such allocations ``exclusive''.%
\footnote{Such allocations are termed ``valid'' in \cite{PTAS_MMS}, or ``at most one'' allocations in \cite{AMO}.} 
%For example, %continuing with the main motivation of \cite{budish2011combinatorial}, when allocating courses to students, one 
This captures applications like course allocation, where a student cannot be allocated multiple seats in the same course. Another example is digital goods with a license quota (e.g., a university with X software licenses that allocates at most one license per student).
%\inbal{removed the third example since it's an example of bads with copies so could be confusing; also don't seem to need more examples.}
%A third example is the task of paper reviewing (it is meaningless to assign the same paper twice to the same reviewer).
Interestingly, even if a ``fair'' allocation (for an appropriate definition of fairness) is guaranteed to hold when each good has a single copy, this may significantly change when multiple copies are available. Considering fair allocation of goods with copies thus raises many open questions.

\noindent{\bf Chores.} 
Fair allocation of chores is a topic of much recent research \cite[e.g.,][]{huang2019choresApproximation}.
%The second model we consider has indivisible {\em chores} (or ``bads''), which are items with negative values. 
The literature usually treats chores separately from goods~\cite[e.g.,][]{IdenticalOrdinalPref,chaudhury2020efx,garg2020improvedMMSgoods}, and our current understanding of chores is lacking. %relative to goods. 
%and as it stands our knowledge on chores is lacking compared to our knowledge on goods. This is due to various reasons, and in particular it seems that the two settings are technically different. 
For example, while it is known that an EFX allocation% 
\footnote{In an EFX (resp.~EF1) allocation, the envy among any two agents can be eliminated by removing any (resp.~some) good from the envied bundle -- see Section~\ref{sec-pre} for formal fairness definitions.}
always exists for any number of goods (without copies) and three agents \cite{chaudhury2020efx}, almost nothing is known regarding EFX allocations of chores, and this is a major open problem. While this difference is a knowledge gap, other gaps provably hold. E.g., for goods, every allocation with \emph{MNW (maximum Nash welfare)} satisfies the ``envy bounded by a single item'' (EF1) fairness notion as well as Pareto efficiency \cite{caragiannis2019unreasonable}. In contrast, for chores, there is no single valued welfare function whose maximization implies both fairness and efficiency \cite{MixedManna}.%
\footnote{This phenomenon persists in the \emph{divisible} case, where there is a polytime algorithm to find a competitive division for goods, but the same problem is PPAD-hard for chores \cite{DBLP:journals/corr/abs-2008-00285}.}

%If we have copies, in fact they are dual! Explain in high level what this means by giving a short survey of our results.

\noindent {\bf Relating the models.} Given the above, one may conclude that studying existence of solution concepts for fair allocation of goods with copies and of chores is to be achieved by two disparate lines of research.
%quests. 
However, a main contribution of this paper is to formalize a tight relation among the models, via a duality theorem that we develop for a large class of fairness notions. The connection between goods and chores has been informally alluded to in the past by \citet{MixedManna}, who give a nice intuitive description: ``Say that we must allocate 5 hours of a painful job [...] Working 2 hours [on the job] is the same as being exempt [from the job] for 3 hours''. Very recently, \citet{PTAS_MMS} employed this connection in their proofs. %We believe that such an important duality should be formally highlighted and explored much beyond a technical tool inside a proof or a short intuitive remark. 
\emph{Our goal in this paper is to develop a formal treatment of this duality in the context of fair allocation existence, in order to enable a comprehensive exploration of the fundamental connection between fair allocations of goods, chores, and their copies.} %We now give an overview of our new findings from this investigation.
% new\inbal{removed the next sentence, it's not clear at this point}
%Furthermore, as we demonstrate, our formal treatment of this duality reveals several conclusions that were, to the best of our knowledge, previously unknown (e.g., the possible inexistence of EFX, the hierarchy of fairness notions that evolve as a result).

%chores can be turned into goods with copies. % \itc{We now describe our results and explain the relation.}
%This kind of idea has appeared in the paper \cite{DBLP:journals/corr/BogomolnaiaMSY17}.
%To the best of our knowledge, it was never studied by rigorous math. 

%\vspace{3mm}
\subsection{Our Results}
We now describe our duality results via a main motivating application. \unless\ifdefined\full{} Throughout, missing proofs appear in the supplementary material.
\fi 

\noindent
\textbf{Main application.}
Consider the well-studied fairness notion of EFX, in which each agent~$i$ prefers her own bundle over the bundle of any other agent~$j$, if one arbitrary good is removed from $j$'s bundle (in the related EF1 notion, the most valued good is removed from $j$'s bundle). As standard in the literature, we assume additive values/costs, and study exclusive allocations.
%``exclusive allocations'' in which an agent can receive at most one copy of each item. 
%
%We begin our exploration with the model of goods with copies.
In light of the recent progress on existence of EFX for goods, culminating with the result of \cite{chaudhury2020efx} for three agents, a natural question is whether EFX existence persists in our first model of interest -- goods \emph{with copies}.
We show (in Section~\ref{sec-warmup-efx}) that an EFX allocation does not always exist for goods with multiple copies. This negative result holds even in a simple setting with any number $n\geq 3$ of agents and identical values (see Example~\ref{ex:GeneralCopiesNoEFX}), demonstrating how adding copies can completely change the fairness landscape.  %(Example~\ref{ex:GeneralCopiesNoEFX}).

To circumvent the underlying reason for inexistence with copies, we introduce a new fairness notion that we term ``EFX without commons'', or ``$\efxwc$'' for short.  
$\efxwc$ requires each agent~$i$ to prefer her own bundle over the bundle of any other agent~$j$, if an arbitrary good \emph{which is not also in $i$'s bundle} is removed from $j$'s bundle. Intuitively, an agent compares herself to her peers while putting aside the intersection of their bundles, i.e., goods common to both.
Where does $\efxwc$ fit in with other solution concepts?
In the special case of a single copy per good, $\efxwc$ coincides with EFX. With multiple copies, $\efxwc$ is a strictly weaker fairness requirement than EFX (i.e., $\efxwc$ is implied by EFX but not vice versa). 
On the other hand, $\efxwc$ is a strictly stronger fairness requirement than EF1.
This immediately raises the question of whether an $\efxwc$ allocation (unlike EFX) is guaranteed to exist for goods with copies.
%weaker fairness notion than EFX, and a stronger one than EF1 (i.e., EFX implies  but not vice versa).
%In light of this inexistence problem, we define ``EFX Without Commons'' ($\efxwc$), which is a weaker fairness notion than EFX, and a stronger one than EF1 (i.e., EFX implies $\efxwc$ but not vice versa).
%Note that a setting without copies is a special case of a setting with copies and exclusive allocations, and in a setting without copies $\efxwc$ is identical to EFX by definition.

At this point, the second model we study -- fair allocation of chores -- becomes relevant. We show that in any fair allocation setting where $n$ is the number of agents, %and any tuple of $n$ valuations/costs, we show that 
an $\efxwc$ allocation for goods with $n-1$ copies per good exists if and only if an EFX allocation for chores %without copies 
exists when the valuations are treated as costs (see Corollary~\ref{cor-efx-for-chores-iff-efxwc}). Thus, we have established a new characterization for the important open problem of existence of EFX for chores. 
This characterization ``justifies'' our new $\efxwc$ notion, since its existence %can be used %as a technical tool 
%to 
determines that of %(standard) 
the well-studied EFX notion for chores. 

%Moreover, the most important technical evidence for the usefulness of $\efxwc$ comes from a surprising connection to the existence of a (standard) EFX allocation for chores. Specifically, for any setting with $n$ agents and any tuple of $n$ valuations/costs, we show that an $\efxwc$ allocation for goods with $n-1$ copies for each good exists if and only if an EFX allocation for chores without copies exists. By this characterization result,
%Thus, existence of $\efxwc$ allocations can be used as a technical tool to determine existence of (standard) EFX allocations for chores. 

\noindent
\textbf{Duality theorems.}
The characterization result for EFX with chores %connection 
is obtained as a corollary of our duality ``meta-theorem'' (Theorem~\ref{thm:wc-duality} in Section~\ref{sec-dual}), connecting goods with copies on the one hand, and chores with copies on the other. This meta-theorem shows, for example, that any allocation of goods with copies is $\efxwc$ if and only if its \emph{dual} allocation for chores with copies is $\efxwc$. 
Since the dual allocation for goods with $n-1$ copies is an allocation with one copy per chore, the characterization result follows. 
The meta-theorem is %designed in an abstract way 
formulated in a general way that fits other envy-based solution concepts as well, and a similar duality meta-theorem (Theorem~\ref{thm-meta-share}) applies to \emph{share-based} solution concepts.%
\footnote{In share-based fairness, an agent compares her allocation to her ``fair share'' of the whole rather than to her peers'; see Section~\ref{sec-pre}.} 
%Our new characterization of existence of EFX allocations for chores is a main corollary of this theorem since the dual allocation for goods with $n-1$ copies is an allocation with one copy of each chore. 
%Our new connection between existence of EFX allocations for chores is a main corollary of this theorem since the dual allocation for goods with $n-1$ copies is an allocation with one copy of each chore. 
The duality meta-theorems are in fact flexible enough to apply to settings with a mixture of goods and chores (see Remark \ref{rem:mixture}).

\noindent
\textbf{Fairness notions ``without commons''.}
In Section~\ref{sec:notions_hierarchy} we further explore ``without commons'' fairness, since its connection to EFX for chores seems to indicate it is the ``right'' notion to study for copies, and since disregarding common items is a natural way to limit envy among agents. We define similar notions to $\efxwc$ for other solution concepts beyond EFX, in particular $\efowc$ and $\eflwc$.%
\footnote{EFL (the basis for $\eflwc$) is due to \cite{barman2018groupwise}.}
We study how these new concepts expand the hierarchy of envy-based fairness notions. Expansion of the hierarchy is important for two reasons: First, intermediate ``without commons'' fairness notions can serve as a technical stepping stone for making progress on the elusive existence of more standard notions. Second, as we show, ``without commons'' envy-freeness has good share-based fairness guarantees,% 
\footnote{As measured by their approximation of MMS, the maximin share of every agent \cite{budish2011combinatorial}.}  
especially in comparison to standard envy-freeness concepts like EF1. %lower in the hierarchy, in particular EF1.
We thus believe such concepts to be of independent interest. 
%much better than $\efo$. MMS is an important share-based (rather than envy-based) fairness notion advocated by Budish .
%we show $\efx \implies \efxwc \implies \efowc \implies \efo$, but not vice versa.
%and present multiple evidence for their usefulness (also see  Figure~\ref{fig:notions_hierarchy}):
%We prove that $\efxwc$ approximates the maximin share (MMS) much better than $\efo$. MMS is an important share-based (rather than envy-based) fairness notion advocated by Budish \shortcite{budish2011combinatorial}.
Our results in this section are summarized in Figure~\ref{fig:notions_hierarchy}. 

%These results are summarized in 
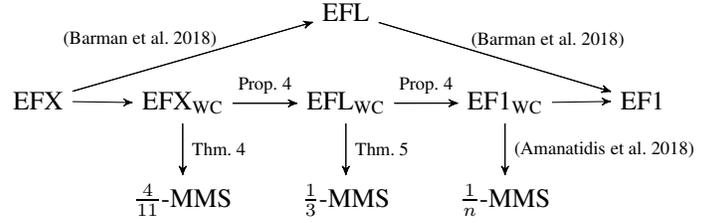
\begin{figure}
\begin{tikzpicture}
  \matrix (m) [matrix of math nodes, row sep=2em, column sep=2em]
    {  &   & \text{EFL}  &   &  \\
      \text{EFX} & \efxwc & \eflwc & \efowc & \text{EF1} \\ 
      & \frac{4}{11}\text{-MMS} &\frac{1}{3}\text{-MMS} & \frac{1}{n}\text{-MMS} & \\};
  { [start chain] \chainin (m-2-1);
    \chainin (m-2-2);
    { [start branch=A] \chainin (m-3-2)
        [join={node[right,labeled] {\text{Thm.~\ref{thm:efx-wc_separation}}}}];}
    \chainin (m-2-3) [join={node[above,labeled] {\text{Prop.~\ref{prop:inclusive_efxwc_eflwc}}}}];
    { [start branch=B] \chainin (m-3-3)
        [join={node[right,labeled] {\text{Thm.~\ref{thm:efl-wc_separation}}}}];}
    \chainin (m-2-4) [join={node[above,labeled] {\text{Prop.~\ref{prop:inclusive_efxwc_eflwc}}}}];
    { [start branch=C] \chainin (m-3-4)
        [join={node[right,labeled] {\text{(Amanatidis et al.~2018)}}}];}
    \chainin (m-2-5);
    }
  { [start chain] 
    \chainin (m-2-1);
    \chainin (m-1-3) [join={node[above,labeled] {\hspace{-40pt}\text{\cite{barman2018groupwise}}}}];
    \chainin (m-2-5) [join={node[above,labeled] {\hspace{40pt}\text{\cite{barman2018groupwise}}}}];
    %\chainin (m-2-5);
    }
\end{tikzpicture}
\caption{The hierarchy among fairness notions for goods with copies, as established in Section~\ref{sec:notions_hierarchy}. A horizontal arrow $A\rightarrow B$ means that an allocation satisfying $A$ also satisfies $B$, but not vice versa. A vertical arrow $A\rightarrow \alpha$-MMS means that an allocation satisfying $A$ guarantees $\alpha$-MMS (in the case of $\efowc$ the guarantee is \emph{at most} $\frac{1}{n}$-MMS).} %\inbal{say the horizontal arrows are ``tight''? and don't we need to say where they're shown (like Prop.~4)?} } \yotam{I don't think we know that $alpha$-MMS ratios are tight. The remaining arrows without theorem/proposition are mentioned in text and not in a formal statement, so I'm not sure how to refer to them}
\label{fig:notions_hierarchy}
\end{figure}

Returning to our main application, while the general question of $\efxwc$ existence for goods with copies is as hard as EFX for chores and thus beyond our current scope, Section~\ref{sec:leveled} provides an affirmative answer for the special case of \emph{leveled} preferences \cite{Babaioff2017CompetitiveEW,ManjunathW21}, where larger bundles are always preferred to smaller ones. This immediately implies the existence of EFX for chores with leveled preferences, thus demonstrating the usefulness of our duality results.

\subsection{Additional Related Literature}

\textbf{Solution concepts.} %Envy-freeness up to a single item ($EF1$) and Maximin share (MMS) as envy-free and share notions respectively were introduced in the seminal paper by \cite{budish2011combinatorial}. 
An EF1 allocation is guaranteed to exist for goods without copies \cite{lipton2004approximately,caragiannis2019unreasonable}, for goods with copies~\cite{cardinalityConstraints2018}, and for chores \cite{mixGoodsChores,EnvyCycleCancellationExample}. Envy-freeness up to any item (EFX) %was introduced by \cite{caragiannis2019unreasonable}, and its 
existence is still a wide open problem for $n\geq4$ agents and additive valuations.
%Recent progress has been made by \cite{chaudhury2020efx}, where they have showed existence for 3 agents in these settings.
An MMS allocation does not always exist, as was observed by Kurokawa \emph{et al}.~\shortcite{alphaMMS}, but fractional approximations of it may be guaranteed, prompting the notion of $\alpha$-MMS. Their work guarantees existence of $\frac{2}{3}$-MMS, and further works guarantee $\frac{3}{4}$-MMS or more for goods, depending on the number of agents \cite{ghodsi2018BestMMSgoods,garg2020improvedMMSgoods}. For chores, Aziz \emph{et al}.~\shortcite{aziz2017MMSChores} and Huang and Lu~\shortcite{huang2019choresApproximation} establish $\frac{11}{9}$-MMS. The envy-free concept of ``EFL'', which we find useful as an intermediary to argue about EFX,  was introduced by Barman \emph{et al}.~\shortcite{barman2018groupwise}, who show its existence for additive goods. For goods, envy-freeness notions can be sorted into a hierarchy (namely, $\text{EFX} \implies \text{EFL} \implies \text{EF1}$), with the notions separated by their $\alpha$-MMS guarantees, as shown by Amanatidis \emph{et al}.~\shortcite{comparingNotions2018}. We are able to establish a similar hierarchy for our more general goods with copies model. 

%\vspace{1mm}

\noindent
\textbf{Restricted allocations.} 
% \itc{Restricted settings.}} 
%\cite{budish2011combinatorial} considers goods with copies in the context of assigning classes to students, but subsequent research, whether of goods or chores, focused on single copies (some justification can be found at \cite{caragiannis2019unreasonable}). 
Kroer and Peysakhovich \shortcite{AMO} consider optimization of MNW by exclusive (``at most one'') allocations, showing these achieve an approximate CEEI in large enough instances. 
%\inbal{commented out next sentence}
%Barman and Krishnamurthy \shortcite{IdenticalOrdinalPref} study identical ordinal preferences for both goods and chores. 
Biswas and Barman~\shortcite{cardinalityConstraints2018} consider a model
%of \inbal{removed ``typed''} goods,
where each agent has cardinality constraints over the amount of a same-typed good she can be allocated. Our model is a special case and thus their general results hold, namely existence of EF1 allocations and of $\frac{1}{3}$-MMS allocations. EFX and EFL allocations may not exist for certain instances in their model as well as in ours. We refine these notions by the introduction of ``without commons'' which enables our duality results, the hierarchy of concepts we describe, and our existence results. Other models with matroid constraints over allocations are presented in \cite{matroidConstraints2019,HeterogenousMatroidConstraints2020}. 
%\inbal{Commented out the next lines}
%Some works assume bundles must be connected in a given underlying graph \cite{fairDivisionGraph,connectedBundles2018}, but our model does not seem to immediately translate to these settings.

%\vspace{1mm}

\noindent
\textbf{Relation between goods and chores.}
%The possible duality between goods and chores was noted in \cite{MixedManna}. They give an example of how to transform a chores instance into a goods instance that seems to follow the transformation outlined in our Definition 8 of the dual allocation. 
%\inbal{moved next sentence here} 
Mixed instances of goods and chores are considered by Aziz \emph{et al}.~\shortcite{mixGoodsChores}. As we show, a mixed model of goods and chores is a special case of our goods with copies model.
Kulkarni \emph{et al.} \shortcite{PTAS_MMS} note that a chore can be reinterpreted as $n-1$ goods, and use this as a technical tool in their PTAS for computing the MMS guarantee in mixed manna settings. %, converting an instance of goods and chores to an optimization problem with goods and ``valid allocations''. 
They do not explore further implications of this duality.
%\inbal{is this ok?}
%\inbal{can we say what the difference from our work is?}

%\vspace{1mm}

%\noindent
%\textbf{Valuation functions:}
%For our $\efxwc$ existence results, we consider special classes of valuation functions, which are prevalent in the general economic literature, and fair allocation literature in particular:
%\begin{itemize}

%    \item \emph{Identical ordinal preferences up to one agent} - \cite{IdenticalOrdinalPref} studied identical ordinal preferences over the goods. We allow a relaxation with one agent not necessarily having the same ordinal preferences. This captures the important case where there is some established global order over the goods, but the details of the valuation may differ among agents. 

%    \item \emph{leveled preferences} - Studied in 
    
    %\item \emph{lexicographic preferences} - Studied in \cite{Babaioff2017CompetitiveEW,lexicographicPref2020}, it is assumed that bundles are compared based on the best good they contain. For additive goods, it captures settings where goods have super-increasing weights (Each good has higher value than the sum of all lower valued goods). 
%\end{itemize}

\section{Preliminaries}\label{sec-pre}
%Let $\types$ denote the set of all types of items. The item of each type appears exactly once in the set $\types$. 
%\vspace{1mm}

In this section we present our model, provide definitions for leading measures of fair allocation of goods, and explain how they generalize to chores.

\noindent
\textbf{Our model.}
Let $\agents$ be a set of $n$ agents and let $\items$ be the set of items to allocate to the agents. 
%We consider environments with indivisible items to allocate to agents.
Let $\types$ be a set of item {\em types} to allocate to the agents
%and let $\tau=|\types|$ %(each item appears once). 
(we slightly abuse notation and use~$\types$ to also denote an item set with a single copy per item). %We will slightly abuse the notation $\types$ to denote the set of items where each type occurs exactly once. % We call the set $\types$ type set. 
The item set $\items$ is a multiset of the type set $\types$, and includes $k_t \le n$ copies of each item type $t\in\types$.
%An item is either a good or a chore.

Every agent $i\in\agents$ has a \emph{valuation} function $v_i :2^{\items}\rightarrow \mathbb{R}$ mapping bundles of values to their values. %\itc{A valuation \emph{profile} $v$ is $(v_1,\dots,v_n)$.} 
An item $t \in \items$ is termed a \emph{good} if $\vai{t}\ge 0$ for each agent~$i$ and $>0$ for at least one agent, and a \emph{chore} if $\vai{t}\le 0$ for each agent~$i$. 
We assume that each item is either a good or a chore, i.e., we do not allow an item to be a good for one agent and a chore for another. We use $t$ when referring to a general item, $g$ when referring to a good, and $c$ when referring to a chore.
Throughout we assume \emph{additive} valuations, i.e., for every bundle of items $S \subseteq \items, v(S) = \sum_{t \in S} v(t)$. Note that additive valuations are \emph{monotone}, i.e., for every two bundles $S_1 \subseteq S_2 \subseteq \items$, we have $v(S_1) \leq v(S_2)$ for goods and 
$v(S_1) \geq v(S_2)$ for chores.

An \emph{allocation} $\allocs=\alloclist$ is a partition of all items among the agents (no item can be left unallocated). We focus on allocations that satisfy the following constraint:

\begin{definition}
An allocation $\allocs$ is {\bf exclusive} if no two copies of the same item are allocated to the same agent.%\footnote{Kulkarni \emph{et al.} \shortcite{PTAS_MMS} termed this ``valid allocations''.}
\end{definition}
(That is, by ``allocation'' we always refer to an exclusive allocation.) We denote the set of all (exclusive) allocations of items $\items$ by $\ea{\items}$. 

%\vspace{1mm}

\begin{definition}
Consider exclusive allocations $\allocs,\allocs'\in\ea{\items}$. Then $\allocs'$ \emph{\bf Pareto dominates} $\allocs$ if for every agent $i$, $\vai{\alloci'}\ge\vai{\alloci}$ and $\vai[j]{\alloci[j]'}>\vai[j]{\alloci[j]}$ for at least one agent $j$.
Allocation $\allocs$ is {\bf Pareto optimal} if there is no $\allocs'$ that Pareto dominates it. 
\end{definition}

\paragraph{Fairness notions for goods and chores.}
The next definition summarizes many prominent fairness notions studied in previous literature, as defined for goods.
%Some fair allocation solution concepts studied in previous literature are as follow.

%\subsection{Standard Fairness Notions} %Envy-based and Share-based}
%Some well accepted fairness notions are as follows.

\begin{definition}\label{def-fair-normal}
An allocation $\allocs\in \ea{\items}$ of \emph{goods} is: 
\begin{itemize}
    \item \emph{\bf EFX} if $\forall i,j\in\agents, \forall g\in\alloci[j] : \vai{\alloci}\ge\vai{\alloci[j]\setminus \{g\}}$;
%For \emph{chores}, an allocation $\allocs$ is $EFX$ if for every $i,j\in\agents$ and item $t\in A_i$,
%$\vai{A_{i}\setminus t}\ge\vai{A_{j}}$.

 \item \emph{\bf EF1} if $\forall\ i,j\in\agents, \exists g\in\alloci[j] : \vai{\alloci}\ge\vai{\alloci[j]\setminus \{g\}}$;%\footnote{Formally speaking, $A_j$ could also be empty.}
 \footnote{In case $\alloci[j]$ is empty, the good exists ``in an empty sense'' and the condition is satisfied.}
     
 \item \emph{\bf EFL} if $\forall\ i,j\in\agents$, either (1) $|A_j| \leq 1,$ or
 (2) $\exists g\in\alloci[j] : \vai{\alloci}\ge\max\{\vai{\alloci[j]\setminus \{g\}},\vai{g}\}$. In words, \emph{EFL} strengthens the \emph{EF1} requirement by not allowing to remove from $A_j$ a highly valuable good, unless $A_j$ is a singleton;
%\end{itemize}

%An exclusive allocation $\allocs$ is 
%\begin{itemize}
    \item \emph{\bf PROP} (proportional) if $\forall\ i \in \agents$ : $\vai{\alloci}\ge\frac{1}{n}\vai{\items}$;
    
    \item \emph{\bf $\alpha$-MMS} (maximin share) for some $0 < \alpha \leq 1$ if
    $\forall i \in \agents$:
%\vspace{-8mm}
    $$\vai{\alloci}\ge \alpha  \max_{\allocs\in\ea{\items}}\min_{j\in\agents}\vai{\alloci[j]}.$$
    \emph{1-\mms} is termed \emph{\mms}.
\end{itemize}
\end{definition}

We can distinguish between two classes of fairness notions appearing in Definition~\ref{def:share-fair}: \emph{share-based} and \emph{envy-based}. Interestingly, this classification determines how the fairness notions are generalized to chores. %We now provide a formal treatment that will play a role in our duality results.

First, PROP and MMS are share-based, meaning that an agent's allocation is measured against what is deemed her fair share:

\begin{definition}
\label{def:share-fair}
Given a valuation $v$, the number of agents $n$ and 
an item set $\items$, a {\bf share} function $s$ outputs a real value, i.e., $s(v,\items, n)\in \mathbb{R}$. An exclusive allocation $\allocs$ is called {\bf $s$-share fair} if we have $\vai{\alloci}\ge s(\vau[i],\items, n), \forall i\in\agents$. 
\end{definition}

\noindent
For example, for PROP, $s(\vau[i],\items, n)=\frac{1}{n}v_i(\items)$. The definitions of PROP and MMS above hold for chores too, and in fact for any mixture of goods and chores. 

The notions of EFX, EF1 and EFL are envy-based, meaning that an agent's allocation is measured against those of her peers. The definition of envy-based notions for chores is different than their definition for goods in that, when an agent $i$ compares her bundle to the bundle of another agent~$j$, she removes a single chore from her \emph{own} bundle (while a good is removed from the other agent's bundle). For example, an exclusive allocation $\allocs$ for chores is EFX if for every agent pair $i,j\in\agents$ and chore $c\in \alloci$:
\begin{equation}
\vai{\alloci\setminus \{c\}}\ge\vai{\alloci[j]}.\label{eq:EFX-chores}
\end{equation}
To more generally phrase the connection between envy-based fairness for goods and for chores we need the following definition.

\begin{definition}[Comparison criterion for goods and chores]
\label{def:comp-criterion}
~~
\begin{itemize}

    \item Given a valuation function $\vau[]$ and two bundles $\bundlei[I]$ and $\bundlei[U]$, a {\bf comparison criterion} $f$ outputs $F$ for ``fair'' and $\neg F$ for ``unfair'', i.e., $f(\vau[],\bundlei[I],\bundlei[U]) \in \{F,\neg F\}$.
    
    \item Given a comparison criterion $f$, the {\bf complementary} comparison criterion $f^c$ satisfies $f^c(\vau[],\bundlei[I],\bundlei[U])=f(-\vau[],\bundlei[U],\bundlei[I])$.
    
\end{itemize}
\end{definition}

For example, let $f$ be the comparison criterion corresponding to EFX for goods (see Definition~\ref{def-fair-normal}), then $f^c$ coincides with the comparison criterion of EFX for chores in Eq.~\eqref{eq:EFX-chores}.
%
% RON: The following explanation might be added to the full version
%For the comparison criterion $f^c$, it will swap the role of bundle $\alloci$ and $\alloci[j]$ and replace $\vau[i]$ by $-\vau[i]$ in the Definition \ref{def-fair-normal}.  So the condition $\vai{\alloci}\ge\vai{\alloci[j]\setminus \{g\}}$ will become $ -\vai{\alloci[j]}\ge-\vai{\alloci\setminus \{g\}}$. This is the condition for chores. If we change the condition from $f^c$ to $f^c_{WC}$, then it is equivalent to Definition \ref{def-efxwc} of $\efxwc$ for chores. 
%
%
One can similarly define EF1 and EFL for chores. More generally: %, we define the following class of envy-based solution concepts for goods and chores.

\begin{definition}[Envy-based fairness for goods and chores]
\label{def-envy-based}
~~
\begin{itemize}

    \item Given a comparison criterion $f$, we say an allocation $\allocs$ is $f$-{\bf fair} iff $\forall i,j\in\agents: f(\vau,\alloci[i],\alloci[j])=F$.

    \item A solution concept \emph{E} is {\bf envy-based} if there exists a comparison criterion $f$ such that an allocation $\allocs$ for goods satisfies \emph{E} iff $\allocs$ is $f$-fair, and an allocation $\allocs$ for chores satisfies \emph{E} iff $\allocs$ is $f^c$-fair.

\end{itemize}
\end{definition}

%\noindent
%EFX, EF1, and EFL are all examples within this general class.

\begin{remark}
\label{rem:mixture}
We define envy-based fair allocation solution concepts for the case where all items are goods and for the case where all items are chores. Our definition -- and in fact all results in this paper -- can be extended to a setting of additive valuations where some items are goods while others are chores, using the envy-based fair allocation solution concepts of \citet{mixGoodsChores}. 
%For brevity we defer this to the full version.
%will not do that.
\end{remark}

%We conclude the section with definitions of a valuation classes we will use in subsequent proofs. Although our focus is additive valuations, we mark in the proofs the precise assumptions used. In the discussion we note the ways our conclusions can thus be extended to more general valuations. 
%\begin{definition}
%A valuation $v$ of goods is 

%\emph{Monotone} - $\forall S_1 \subseteq S_2 \subseteq \items, v(S_1) \leq v(S_2)$. 

%\emph{Subadditive} - $\forall S_1,S_2 \subseteq \items, v(S_1) + v(S_2) \geq v(S_1 \cup S_2)$.

%\emph{Supperadditive} - $\forall S_1,S_2 \subseteq \items, v(S_1) + v(S_2) \leq v(S_1 \cup S_2)$.

%\emph{Additive} - a valuation that is both subadditive and superadditive. 
%\end{definition}

%\section{Motivation / the need / our framework for EFX}
\section{EFX for Goods with Copies}
\label{sec-warmup-efx}

We begin by showing that EFX may not exist with copies.

%\subsection{EFX inexistence}

\begin{example}
\label{ex:GeneralCopiesNoEFX}
Consider $n\geq 3$ agents, $n+1$ goods, and $\frac{n}{2}<k<n$ copies of each good. All agents share the same additive valuation function~$v$, and for every item type $g_w\in \items$: 
$$
%\forall 1 \leq i \leq n~:~
v(g_w) = \begin{cases} 
w & w < n+1; \\
n^2 & w = n+1.
\end{cases}
$$
\end{example}

\begin{prop}
\label{pro:GeneralCopiesNoEFX}
For goods with copies, an \emph{$\efx$} allocation may not exist even for 3 agents. The same holds for~\emph{$\efl$}.
\end{prop}

%\ifdefined\full

\begin{comment}
\begin{proof}
Consider Example~\ref{ex:GeneralCopiesNoEFX}. Fix an agent $l$ who does not receive a copy of $g_{n+1}$. %Let $h$ be an agent contains both goods $g_1$ and $g_{n+1}$. 
Since there are $k$ copies of $g_1$ and $k$ copies of $g_{n+1}$ where $2k>n$, there exists an agent $h$ which receives both $g_1$ and $g_{n+1}$. Agent $l$ EFX envies agent $h$ since
$$
\vai[l]{\alloci[l]} \le \sum_{w\le n}\vai[l]{g_{w}} < \vai[l]{g_{n+1}} \le \vai[l]{\alloci[h]\setminus\{g_1\}}.
$$

%\vspace{-8mm}
\end{proof}
\end{comment}

\begin{proof}
Consider Example~\ref{ex:GeneralCopiesNoEFX}. We show non-existence of EFL allocations. Since every EFX allocation is also EFL \cite{barman2018groupwise}, this shows no EFX allocations exist for these settings. Fix an agent $l$ who does not receive a copy of $g_{n+1}$. %Let $h$ be an agent contains both goods $g_1$ and $g_{n+1}$. 
Since there are $k$ copies of $g_1$ and $k$ copies of $g_{n+1}$ where $2k>n$, there exists an agent $h$ which receives both $g_1$ and $g_{n+1}$. Agent $l$ EFL-envies agent $h$ since the first condition of EFL is not satisfied ($|\alloci[h]| \geq 2$), and for any $g \in \alloci[h]$ with $\vai[l]{g} \leq \vai[l]{\alloci[l]}$,
$$
\vai[l]{\alloci[l]} \le \sum_{w\le n}\vai[l]{g_{w}} < \vai[l]{g_{n+1}} \le \vai[l]{\alloci[h]\setminus\{g\}},
$$
violating the second condition of EFL. 
\end{proof}

%\begin{equation*}
%\begin{split}
%\vai[l]{\alloci[l]}& \le \sum_{w\le %n}\vai[l]{g_{w}} \\
% & < \vai[l]{g_{n+1}}\\
% &\le \vai[l]{\alloci[h]\setminus\{g_1\}}
%\end{split}
%\end{equation*}
%$\vai[l]{\alloci[l]}\le \sum$

%Let $L$ be the set of agents that do not receive the most valuable good $t_{n+1}$. It is easy to see the size $|L|=n-k$. and $W$ the set of all other agents that receive it. Let $l\in L$ be some agent of the first set. If any $w\in W$ has some additional good $g$ besides $n+1$, then $$v(\alloci[w] \setminus \{g\}) \geq v(n+1) = n^2 > \frac{n(n+1)}{2} \geq v(\alloci[l]),$$
%and so $l$ $EFX$-envies $w$. But there must be some such $w$ by the pigeonhole principle.  
%\fi

%\subsection{EFX ``without commons''}
%\subsection{Envy freeness ``without commons''}
%Introduce the notion and give an example.  Though we add a constraint, it is a more general notion. The notion of $\efxwc$ is equivalent to EFX if there is only one copy of each good. We should define $\efxwc$ both for goods with copies and chores with copies, but emphasize that after showing duality we will focus on the model of goods with copies.

To handle the non-existence problem we introduce the new fairness notion of $\efxwc$ -- a generalization of EFX to a model with copies. When comparing two bundles, $\efxwc$ puts aside all items common to both bundles, and then compares the remaining items in an EFX way:

%Given an allocation $\allocs$, we will use the notation $\alloci\setminus\alloci[j]$ to denote the items that appears in the bundle $\alloci$ but not appears in the bundle $\alloci[j]$.

\begin{definition}[$\efxwc$ for goods]
\label{def-efxwc}
An exclusive allocation $\allocs$ of goods is \emph{$\efxwc$} if for every two agents $i,j\in\agents$ and item $g\in \alloci[j]\setminus\alloci$ it holds that
$
\vai{\alloci\setminus\alloci[j]}\ge\vai{(\alloci[j]\setminus\alloci)\setminus \{g\}}.
$
\end{definition}

\noindent
\begin{remark}
We remark that:
~~
\begin{itemize}
    \item \emph{$\efxwc$} is an envy-based notion and therefore its definition extends to chores using Definition~\ref{def-envy-based}.

    \item A model with no copies is a special case of our model, and when there is one copy of each good in our model then \emph{$\efxwc$} is equivalent to \emph{EFX.}
%can be seen as a special case of $\efxwc$ when there is only one copy of each item. Besides these, this notion of goods and chores has a duality relationship in our model.

    \item \emph{EFX} implies \emph{$\efxwc$}: By additivity of the values, the inequality in Def.~\ref{def-efxwc} is equivalent to the EFX inequality $\vai{\alloci}\ge\vai{\alloci[j]\setminus \{g\}}$, but \emph{EFX} allows $g$ to be any item while \emph{$\efxwc$} restricts its choice.

\end{itemize}
\end{remark}

A main technical justification of $\efxwc$ is its ability to give a ``dual'' view of goods and chores, which also yields a new characterization of existence of standard EFX allocations for chores. This is summarized by the following result (it is a corollary of Theorem~\ref{thm:wc-duality} stated in the next section). %(including the notion of a dual allocation).

\begin{corollary}[Characterization of EFX existence for chores]
\label{cor-efx-for-chores-iff-efxwc}
%An exclusive allocation is $\efxwc$ for goods iff its dual allocation is $\efxwc$ for chores. As an important application, 
An \emph{EFX} allocation for chores exists in the standard setting without copies iff an \emph{$\efxwc$} allocation for goods exists in our setting with $k=n-1$ copies of each good.
%As an important special case we obtain that, for any arbitrary setting $n$, $v_1(\cdot),...,v_n(\cdot)$, and $\types$,
%an EFX allocation for chores in this setting exists iff an $\efxwc$ allocation for goods with $k=n-1$ copies of each $g \in \types$ in this setting exists. A similar statement can be made regarding a mixed manna of goods and chores.
\end{corollary}

The following example illustrates the above.
\begin{example}
[Special case of Example~\ref{ex:GeneralCopiesNoEFX}]
\label{ex-dual}
There are $n=3$ agents and 4 goods with $n-1=2$ copies each. All agents have valuation $v$ as defined in Example \ref{ex:GeneralCopiesNoEFX} (the values of items $t_1,t_2,t_3,t_4$ are $1,2,3,9$ respectively).
\end{example}

In Example~\ref{ex-dual}, the allocation $\alloci[1]=\{t_1,t_2,t_4\}$, $\alloci[2]=\{t_3,t_4\}$ and  $\alloci[3]=\{t_1,t_2,t_3\}$ is $\efxwc$. Indeed, agent~3 EFX-envies the others (and moreover EFL-envies them), 
but after putting aside common items, only item $t_4$ is left for the other agents, implying $\efxwc$.
%Let's take the dual of this allocation (the formal definition is in Section \ref{sec-dual}). %
We now show that the ``dual'' allocation is $\efx$ for chores. The dual allocation is given by $\dual{\alloci}=\types\setminus\alloci$ for every $i$ (see formal definition in Section~\ref{sec-dual}). I.e., $\dual{\alloci[1]}=\{t_3\}$, $\dual{\alloci[2]}=\{t_1,t_2\}$ and $\dual{\alloci[3]}=\{t_4\}$. When the goods are treated as chores by simply taking the negation $(-v)$ of the valuation (such that the values of $t_1,t_2,t_3,t_4$ are $-1,-2,-3,-9$ respectively), %valuation for chores corresponding to valuation $v$ for goods is $\dual{v}=-v$. 
it is not hard to check that the dual allocation is EFX for these chores. Indeed, agent~3 envies the others, but removing any chore from agent~3's single-chore bundle alleviates the envy. %(recall Eq.~\eqref{eq:EFX-chores}).}

\section{Fairness Duality for Copies and Chores}\label{sec-dual}
%In this section, we will introduce a duality transformation between goods and chores on our model. In our model, we show that to find a fair allocation for chores  is equivalent to find a fair allocation on goods. 

In this section we show that allocations of goods and of chores are ``dual'' in a formal sense, and more importantly that fairness notions translate between dual allocations. %(Sec.~\ref{sub:meta-envy} shows this for envy-based and Sec.~\ref{sub:meta-share} for share-based).

\begin{definition}[Duality]
%We introduce a dual transformation of our setting:
The {\bf dual} of a tuple $(\allocs,v,\items)$ is:
\begin{itemize}
   % \item The dual of an item $t\in\types$ is $\dual{t}$. The dual type set $\dual{\types}=\bigcup_{t\in\types}\dual{t}$.
    \item The dual allocation $\dual{\allocs}$ is $\dual{\alloci}=\types\setminus\alloci$.   
    
    \item The dual valuation %\itc{profile} 
    is $\dual{v}=-v$. Thus goods become chores and vice versa. %\itc{for every agent}.

    \item The dual item set $\dual{\items}$ 
    contains $n-k_t$ copies of every $t \in \types$, where $k_t$ is the number of copies of $t$ in $\items$.

\end{itemize}

\end{definition}
%Notice that when we mention the dual, we will dual the allocation, valuation and item set at the same time.

%\noindent
%A dual transformation converts goods to chores and vice versa. 
%by the valuation sign reversal, the dual of a goods with copies instance is a chores with copies instance, and vice versa. 

\begin{prop}[Properties of dual transformations]
\label{prop:duality}
~~
\begin{enumerate}
    \item If $\allocs$ is an exclusive allocation of item set $\items$ then $\dual{\allocs}$ is an exclusive allocation of item set $\dual{\items}$.
    
    \item The dual of the dual is the original: $\dual{\dual{\allocs}}=\allocs$, $\dual{\dual{v}}=v$ and $\dual{\dual{\items}}=\items$.
    
    \item The dual operation is a one-to-one mapping.
\end{enumerate}
\end{prop}

%\ifdefined\full 
\begin{proof}
\begin{enumerate}
    \item $\dual{\allocs}$ is an allocation since each item $t$ is allocated $k_t$ times in $\allocs$, so it does not appear in the bundles of $n-k_t$ agents. Every such agent is allocated a copy of $t$ in $\dual{\allocs}$, and other agents are not.    $\dual{\allocs}$ is exclusive since by construction each bundle $\dual{\alloci}$ is contained in $\types$ and has at most one copy of each item.

    \item $\dual{\dual{\alloci[i]}} = \types \setminus \dual{\alloci[i]} = \types \setminus (\types \setminus \alloci[i]) = \alloci[i],\
    \dual{\dual{v}} = -\dual{v} = v$,
    and $\dual{\dual{\items}}$ contains $n-(n-k_t) = k_t$ items $\forall t \in \types$. 
    
    \item Assume that $\allocs_1, \allocs_2, v_1, v_2, \items_1, \items_2$ satisfy $\dual{\allocs_1} = \dual{\allocs_2}$, $\dual{v_1} = \dual{v_2}$, $\dual{\items_1} = \dual{\items_2}$, then we have $\allocs_1 = \dual{\dual{\allocs_1}} = \dual{\dual{\allocs_2}} = \allocs_2$. The rest follows similarly.
\end{enumerate}
\vspace{-4mm}
\end{proof}
%\fi 

\begin{theorem}
An exclusive allocation $\allocs$ is Pareto optimal iff the dual allocation $\dual{\allocs}$ is Pareto optimal. 
\end{theorem}
\begin{proof}
First we prove a property of the dual: If $\vai{\bundlei}>\vai{\alloci}$ then $\dual{\vau[i]}(\dual{\bundlei})>\dual{\vau[i]}(\dual{\alloci})$. From $\vai{\bundlei}>\vai{\alloci}$, we have $\vai{\types\setminus\bundlei}<\vai{\types\setminus\alloci}$, i.e.  $\vai{\dual{\bundlei}}<\vai{\dual{\alloci}}$. Because $\dual{\vau[i]}=-\vau[i]$, we have $\dual{\vau[i]}(\dual{\bundlei})>\dual{\vau[i]}(\dual{\alloci})$.

%We prove the ``if" part. Suppose that the allocation $\dual{\allocs}$ is Pareto optimal. We prove it by contradiction. Suppose there is an allocation $\bundels$  Pareto dominating $\allocs$

Suppose that allocation $\bundles$ Pareto dominates allocation $\allocs$. We prove that allocation $\dual{\bundles}$ Pareto dominates allocation $\dual{\allocs}$. %Pareto dominating means that $\vai{\bundlei}\ge\vai{\alloci}$ for all $i\in\agents$ and there is a $j$ such that $\vai[j]{\bundlei[j]}>\vai[j]{\alloci[j]}$. 
By the property we just proved, we have $\dual{\vau[i]}(\bundlei)\ge\dual{\vau[i]}(\alloci)$ for all $i\in\agents$ and there is a $j$ such that $\dual{\vau[j]}(\dual{\bundlei[j]})>\dual{\vau[j]}(\dual{\alloci[j]})$
This directly implies the statement, as there is a Pareto dominating allocation over $\allocs$ iff there is such an allocation for the dual. 
\end{proof}

\subsection{A Meta-theorem for Envy-based Notions}
\label{sub:meta-envy}

The idea of $\efxwc$ generalizes to other envy-based notions:

\begin{definition}[Comparison without commons]
\label{def-wc}
 Given a comparison criterion $f$, the comparison criterion \emph{$f_{\text{WC}}$} satisfies $f_{\text{\emph{WC}}}(\vau[],\bundlei[I],\bundlei[U])=f(\vau[],\bundlei[I]\setminus\bundlei[U],\bundlei[U]\setminus\bundlei[I])$, for all $\vau[],\bundlei[I],\bundlei[U]$.
\end{definition}

\noindent
As additional examples to $\efxwc$, consider $\efowc$ and $\eflwc$.
Note that by Definition~\ref{def:comp-criterion}, for any comparison criterion~$f$ we have $(f_{\text{WC}})^c = (f^c)_{\text{WC}} \equiv f^c_{\text{WC}}$.  
We now show our first main theorem -- that envy-based fairness holds under a duality transformation from goods to chores.%\begin{remark}
%\footnote{
%Notice that the definitions of EFX, EFL, EF1 do not require the valuation to be additive. We can define EFX, EFL, EF1 on any monotone valuation without changing anything in Definition~\ref{def-fair-normal}. Our proof and definitions for envy-based fairness notions also do not require the valuation to be additive. Here our result can extend to any envy-based notions which is defined on any possible valuations (e.g. submodular, superadditive, monotone valuations).}%\end{remark}
%In the following sections we discuss the existence of allocations that satisfy these notions and the hierarchy that they create. Here, we show how they capture the duality of allocations for goods and chores.

\begin{theorem}
\label{thm:wc-duality}
Given a comparison criterion $f$, an exclusive allocation $\allocs$ is $f_{\text{\emph{WC}}}$-fair with respect to $v$ iff its dual $\dual{\allocs}$ is $f_{\text{\emph{WC}}}^c$-fair with respect to $\dual{v}$.
\end{theorem}
\begin{proof}
We prove $f_{\text{WC}}(\vau,\alloci,\alloci[j])=f^c_{\text{WC}}(\dual{\vau},\dual{\alloci},\dual{\alloci[j]})$ for all $i,j\in\agents$. Let $O_i=\alloci\setminus\alloci[j]$ and $O_j=\alloci[j]\setminus\alloci$. Note that $O_i=\dual{\alloci[j]}\setminus\dual{\alloci}$ and $O_j=\dual{\alloci}\setminus\dual{\alloci[j]}$. Thus,
\[
\begin{split}
f_{\text{WC}}&(\vau,\alloci,\alloci[j])=f(\vau,O_i,O_j)
 =f^c(-\vau,O_j,O_i)\\
 &=f^c_{\text{WC}}(\dual{\vau},\dual{\alloci},\dual{\alloci[j]}).
\end{split}
\]
%\vspace{-4mm}
%Let us look at agent $i$ and bundles $\alloci$ and $\alloci[j]$.
\end{proof}

%for the definition of fairness notions for mixed goods and chores. We can further extend the result to transform mixed goods and chores to pure goods.
%Notice that in our setting, we do not allow an item can be goods for some agents and chores for other agents, which is slightly different from the work \cite{mixGoodsChores}.

%\yotam{That's a really good point. Our original proof assumed additivity, but it ignored that we cancel duplicated items again after applying the dual. EFX etc are defined even for non-monotone valuations, no? Just it doesn't make much sense then conceptually.}

It follows from Theorem~\ref{thm:wc-duality} that an exclusive allocation is $\efxwc$ for goods if and only if its dual allocation is $\efxwc$ for chores. An important application is
Corollary~\ref{cor-efx-for-chores-iff-efxwc} in Section~\ref{sec-warmup-efx}, since the dual of a setting of chores with no copies is a setting of goods with $n-1$ copies of each good, and in a setting of chores with no copies $\efxwc$ is identical to EFX.

%An exclusive allocation is $\efxwc$ for goods iff its dual allocation is $\efxwc$ for chores. As an important application, 
%An \emph{EFX} allocation for chores exists in the standard setting without copies iff an \emph{$\efxwc$} allocation for goods exists in our setting with $k=n-1$ copies of each good.

\subsection{A Meta-theorem for Share-based Notions}
\label{sub:meta-share}

Does a similar duality result hold for share-based fairness? We give an affirmative answer for a class of share-based notions that satisfy a property we term \emph{linear shares}.  
As we show below, MMS and PROP are both notions in this class. %(see Proposition~\ref{pro:who-is-linear} below).
The idea of a linear shares is as follows. 
Focus on one agent and consider all values from her perspective. 
For any allocation, our dual transformation shifts the value of every bundle in her eyes by a constant $d$ (in particular, $d=v(\types)$). 
It is natural to require that her fair share for this allocation also shifts by the same constant $d$. A linear share-based notion is one for which this property holds. The following definition formalizes this intuition, while generalizing the dual transformation to appropriate one-to-one mappings.

%For example, PROP satisfies this since if every bundle is say added a copy of an item worth $d$, then of course the proportional share grows by $d$. 
%The idea of the linear share notion is that if we are guaranteed that for all allocations, the valuation of every bundle shifts by a constant $d$, then we expect the share notion to shift in the same constant $d$.

%under duality, the valuation of every bundle shifts by a constant (in particular $v(\types)$). 

%For example, PROP satisfies this since if every bundle is say added a copy of an item worth $d$, then of course the proportional share grows by $d$. 

%This property means that if we shift all valuations by some constant $d$, the share also shifts by $d$. 

%Intuitively, for any valuation function $v$, if we add $n$ copies of an item $t$ to the item set, then the share should increase $v(t)$.
%\begin{definition}
%A share is called \emph{linear} if for any item set $\items$, an extra item $t$ and valuation function $v$, $s(v,n,\items\cup\{t\}\times n)=s(v,n,\items)+v(t)$.
%\end{definition}

\begin{definition}[Linear shares]
\label{def:linear-share}
Two item sets $\items$ and $\items'$ are {\bf linearly related} w.r.t.~valuations $\vau[]$, $\vau[]'$ and a real constant~$d$ if there is a one-to-one mapping from exclusive allocations $\allocs\in\ea{\items}$ to exclusive allocations $\allocs'\in\ea{\items'}$  such that $\vai[]{\alloci}=v'(\alloci')+d$ for all $i\in\agents$. %and all allocations.
A share $s$ is {\bf linear} if for any item sets $\items$ and $\items'$ that are linearly related w.r.t.~$\vau[]$, $\vau[]'$ and $d$ we have $s(v,\items, n)=s(v',\items', n)+d$. %, where $0\le b\le 1$ is some entitlement\footnote{Shares with entitlements generalize the concept of shares in a way that accommodates AnyPrice Share and Truncated Proportional Share. These notions can have arbitrary entitlements, while PROP  and MMS  require $b=\frac{1}{n}$ as a parameter.}. 
\end{definition}

An example of linearly related item sets and valuations is $\items,\items',v,v'$ where $\items'$ is equal to $\items$ with an additional $n$ copies of a new item $t$, and $v'$ is equal to $v$ with an additional value for $t$. Thus if $s$ is a linear share %(\itc{e.g.~PROP}) 
then $s(v,\items, n)=s(v',\items', n)-v'(t)$. %This is certainly true if we think $s$ is . 

\begin{theorem}\label{thm-meta-share}
For any linear share $s$, an exclusive allocation $\allocs$ is $s$-share fair for $\items, v$ iff $\dual{\allocs}$ is $s$-share fair for $\dual{\items}, \dual{v}$.
\end{theorem}

\begin{proof}
We prove that for every agent $i\in\agents$, item sets $\items$ and $\dual{\items}$ are linearly related w.r.t.~valuations $\vau$, $\dual{\vau}$ and $d=\vai{\types}$. Given an allocation $\allocs\in\ea{\items}$, let us consider the dual transformation, which is a one-to-one mapping that satisfies the condition in Definition~\ref{def:linear-share}:
\begin{equation*}
\begin{split}
\vai{\alloci[j]}- & \dual{\vau}(\dual{\alloci[j]}) =\vai{\alloci[j]} -(-\vai{\dual{\alloci[j]}}) \\
 & = \vai{\alloci[j]}+\vai{\types\setminus\alloci[j]} = \vai{\types}.
\end{split}
\end{equation*}

Suppose that $\allocs$ is $s$-share fair for $\items, v_i$. By definition, we have $\vai{\alloci}\ge s(\vau,\items, n)$ for all $i\in\agents$. Since the share $s$ is linear, $s(\vau,\items, n)=s(\dual{\vau},\dual{\items}, n)+\vai{\types}$. Therefore, $\dual{\vau}(\dual{\alloci})\ge s(\dual{\vau},\dual{\items}, n)\ \forall i\in\agents$, and $\dual{\allocs}$ is $s$-share fair for $\dual{\items}, \dual{v_i}$. 
The converse is proved similarly.
 \end{proof}

\ifdefined\aps{}
The linear share is a sufficient condition for the duality. We notice that a new notion called AnyPrice share \cite{fairshare_arbitraryentitlements} can have a similar duality without linearity. We discuss this in an appendix in the full version.
 \fi
% \begin{remark} In Theorem \ref{thm-meta-share}, the allocation $\allocs$ and item set $\items$ can have both goods and chores, which means it can be applied for mixed setting. And we can prove a stronger statement that for any mixed setting,  it is $s$-share-fair if and only if a pure goods instance is $s$-share-fair after dual all chores. Since it requires new definitions for duality, we omit it here.
% \end{remark}

%\yotam{To include APS and TPS in the discussion, do we need to add the notion of shares with entitlements? Is there a more simple way?}

\noindent{\bf Examples and non-examples of linear shares.} Linearity is an arguably natural property and indeed holds for the most prominent share-based notions. 

\begin{prop}
\label{pro:who-is-linear}
\emph{MMS} and \emph{PROP} are linear.
\end{prop}

\ifdefined\full
\begin{proof}
Suppose that item sets $\items$ and $\items'$ are  linearly related with $v$, $v'$ and $d$.
We have a one to one mapping from exclusive allocations $\allocs\in\ea{\items}$ to exclusive allocations $\allocs'\in\ea{\items'}$  such that $\vai[]{\alloci}=v'(\alloci')+d$ for all $i\in\agents$ and all allocations. 

For MMS, we have %$\max_{\allocs\in\ea{\items}}\min_{j\in\agents}\vai{\alloci[j]}=\max_{\allocs\in\ea{\items'}}\min_{j\in\agents}\vau'(\alloci[j])+d=d+$
\begin{equation*}
\begin{split}
\max_{\allocs\in\ea{\items}}\min_{j\in\agents}\vai{\alloci[j]} &=\max_{\allocs\in\ea{\items'}}\min_{j\in\agents}\vau'(\alloci[j])+d\\
&=d+\max_{\allocs\in\ea{\items'}}\min_{j\in\agents}\vau'(\alloci[j])
\end{split}
\end{equation*}

For PROP, it holds that $\vai{\items} = \sum_{j\in\agents} \vai{\alloci[j]} = \sum_{j\in\agents} (\vau'(\alloci[j]') + d) = \vau'(\items')+n\cdot d$. %if very bundle increases a constant $d$, then total value must increase $n\cdot d$. 
So the value of PROP increases by $d$.
\end{proof}
\fi

Interestingly, more nuanced share-based notions do not always satisfy linearity, thus portraying the possible limits of the duality between fair allocations of copies and chores. %has limits. 
%In contrast to MMS and PROP, 
For example, $\alpha$-MMS is not linear, since a shift $d$ in the valuation translates to a shift $\alpha d$ in the share. %To give a better understanding of this concept of a linear share, we examine 
\unless\ifdefined\full{}
As we show in the full version, this is also the case for the new notion of \emph{truncated proportional share} (TPS), an extension of PROP
\fi
\ifdefined\full{}
We show this is also the case for the new notion of
\fi
\emph{truncated proportional share} (TPS), which is defined for goods in \cite{fairshare_arbitraryentitlements}.\footnote{\citet{fairshare_arbitraryentitlements} propose another interesting notion called \emph{AnyPrice share}, but extending it to the setting of copies seems challenging so we do not address it here.} %. Therefore, we do not compare AnyPrice share here.} 
%\inbal{I think the rest should go to the supplementary material.} 
\ifdefined\full{}
The notion is an extension of PROP and it has been shown that $\text{PROP}\ge \text{TPS}\ge \text{MMS}$.

%\xin{We should define price (or equivalent probability) clearly here. Is it better to define goods and chores separately?}
\begin{definition}
(Truncated Proportional Share for goods) Given an instance and an agent $i$,  the truncated proportional share is
%Given a valuation $v$, the TPS 

    TPS($v_i,\items,n$)=$\max_z \{ z | \frac{1}{n} \sum_{g\in \items} \min\{\vai{g},z\} = z\}.$

%    TPS for chores: $\min \{ \frac{1}{n}\cdot v(\items), \min_{i\in \items} \{ v(i)\} \}$

\end{definition}

%For chores,  this preserves the intuition of TPS for goods, where large items prevent an equal distribution of goods, and so need to be truncated for a realistic estimation of a possible fair share. In the chores case, large items must be assigned, and so a fair share must consider receiving them as a singleton bundle. This definition preserves the important property that $TPS \leq PROP$. Yet, it does not satisfy duality, even for the setting of Example~\ref{ex:APS}, for reasonable choice of $b$. 

The following example shows that TPS is not a linear share:

\begin{example}
Suppose that there are 3 agents and 3 goods. Let us focus on one agent's valuation: 2, 3, 5. The TPS would be 2 at this point. If we add 3 copies of a good with  valuation 3, then PROP would be $\frac{19}{3}>5$. After adding these good copies, the TPS becomes $\frac{19}{3}$, which increases more than 3. 
\end{example}

\ifdefined\extraTPS{}
We wish to further show that this does not just violate the linearity but also the duality. For this, we first have to define TPS for chores. 

\begin{definition}
(Truncated Proportional Share for chores) Given an instance and an agent $i$,  the truncated proportional share is

    TPS($v_i, \items, n$)=$\min \{ \frac{1}{n}\cdot v(\items), \min_{i\in \items} \{ v(i)\} \}$

\end{definition}

To the best of our knowledge, TPS was not defined for chores previously, so this definition requires some justification. First, it preserves the intuition of TPS for goods, where since large items prevent an equal distribution of goods, they are truncated for a more realistic estimation of the possible fair share. In the chores case, large items must be assigned, and so a fair share must consider receiving them as a singleton bundle. The definition also preserves the important property that $TPS \leq PROP$. Yet, it does not satisfy duality, as we see in the following example:

\begin{example}
\label{ex:TPS}
Consider $n=4$ agents, 6 goods with 1 copy each, and valuations $1,3,4,6,7,19$ respectively. It has $PROP = 10, TPS = 7$ as $\frac{1}{4} \sum_{g\in \items} \min\{v(g),7\} = \frac{1}{4} (1 + 3 + 4 + 6 + 7 + 7) = 7$, but for any $z > 7$, $\frac{1}{4} \sum_{g\in \items} \min\{v(g),z\} = \frac{1}{4} (1 + 3 + 4 + 6 + 7 + z) = 7 + \frac{z - 7}{4} \neq 7 + z - 7 = z$. The TPS value is implemented by allocation $\allocs = \{\{1,6\}, \{3,4\}, \{7\}, \{19\}\}$. 
The dual chores instance has $6$ chores with three copies each and valuations $-1, -3, -4, -6, -7, -19$ respectively. For this instance, $PROP = TPS = -30$ by direct calculation. But the dual allocation $\dual{\allocs}$ %$\dual{\allocs} = \{\{-3,-4,-7,-19\}, \{-1,-6,-7,-19\},\{-1,-3,-4,-6,-19\},\{-1,-3,-4,-6,-7\}\}$
has a bundle with value $-33$, which is less than the TPS value. 
\end{example}

\fi

\fi

%\noindent
%The proof for Proposition.~\ref{pro:who-is-linear}, as well as all other redacted proofs, are given in the supplementary extended version. 

%\yotam{I'm not sure it's possible to define TPS for chores in a meaningful way. Consider a tweak of Example 5 in \href{https://arxiv.org/pdf/2102.04909.pdf}{https://arxiv.org/pdf/2102.04909.pdf}. 4 agents, 5 chores, valued -2,...,-6. We have $PS = -5$. For $TPS \leq -5$ (since we expect $APS \leq TPS \leq PS$), we need to cap $-2,...-5$, so we have at most $\frac{-5 \cdot 4 - 6}{4} < -6$, it turns out that there is no solution to the equation. For this case our APS definition for chores would be $-6$. }

%\subsection{Duality for Mixed Goods and Chores}

\section{Expanding the Envy-Based Hierarchy}
\label{sec:notions_hierarchy}

Given the duality between goods with copies and chores with copies, we now focus on the former. We first analyze the connections between the new envy-based fairness notions and the more standard ones, as summarized in Figure~\ref{fig:notions_hierarchy}.
In Section~\ref{sub:MMS} we show that this hierarchical structure implies strictly different approximation guarantees to MMS.

%\subsection{Inclusive Relationship}
%It is possible to establish a hierarchy of the different envy-freeness notions for goods with copies. Similar to goods (without copies), the fairness notions satisfy $\efxwc \stackrel{(1)}{\rightarrow} \eflwc \stackrel{(2)}{\rightarrow} \efowc$, but not vice versa. We prove the first implication by the following lemma and example.

%\yotam{Monotone should be enough}
\begin{prop} %\textcolor{red}{Add words (?)}
\label{prop:inclusive_efxwc_eflwc} 
$\efxwc \implies \eflwc \implies \efowc$.
\end{prop}
\ifdefined\full
\begin{proof}
($\efxwc \implies \eflwc$)
Fix an $\efxwc$ allocation $\allocs$ and $i,j \in \agents$. We wish to show $j$ does not $\eflwc$-envy $i$. If $|\alloci[i] \setminus \alloci[j]|\leq 1$, the first $\eflwc$ condition holds. Thus assume that $|\alloci[i] \setminus \alloci[j]|\geq 2$. Suppose that there is a good $g \in \alloci[i] \setminus \alloci[j]$ s.t.~$v_i(g) >  v_i(\alloci[j] \setminus \alloci[i])$. Since there are at least 2 goods, there must be another good $g^* \in (\alloci \setminus \alloci[j]) \setminus g$. For $g^*$, we have

$$v_i((\alloci \setminus \alloci[j]) \setminus g^*) \stackrel{monotone}{\geq} v_i(g) >  v_i(\alloci[j] \setminus \alloci[i]),$$ violating the $\efxwc$ condition. 
%the $\efxwc$ condition fails when removing a good different than $g$.
We conclude that all goods $g \in \alloci \setminus \alloci[j]$ satisfy $v_i(g) \leq  v_i(\alloci[j] \setminus \alloci[i])$. Again by the $\efxwc$ condition, $v_i((\alloci[i] \setminus \alloci[j]) \setminus \{g\})) \leq  v_i(\alloci[j] \setminus \alloci[i])$. We can therefore choose an arbitrary good and the second $\eflwc$ condition holds. 

($\eflwc \implies \efowc$) The only difference occurs when the second $\eflwc$ condition is invoked. Then, there is an extra constraint over the choice of good $g$ for $\eflwc$ condition that does not appear in the $\efowc$ definition. Therefore, the $\eflwc$ definition is stricter than the $\efowc$, and the implication follows. 
\end{proof}
\fi

\begin{example}($\eflwc \centernot \implies \efxwc$) Consider two agents and five goods with a single copy each, and identical valuations $a = b = 1, c = d = 1+\epsilon, e = \epsilon = 0.01$. The allocation $\{\{c,d,e\},\{a,b\}\}$ is $\eflwc$ but not $\efxwc$. %Agent 2 envies agent 1, but removing $e$ from agent 1 doesn't suffice to eliminate the envy (Not $EFX$). On the other hand, $v_2(c) < v_2(\alloci[2])$, and $v_2(\alloci[1] \setminus \{c\}) < v_2(\alloci[2])$ ($EFL$). 
\end{example}

%We are also able to show that the notion of EFL is incomparable to the three $f_{wc}$ notions we defined. By the $f_{wc}$ notions hierarchy, it is enough to show that 

\begin{prop}
$\efxwc \centernot\implies \text{EFL}, \text{EFL}\centernot\implies \efowc$. 
\end{prop}
The first negation is due to Example~\ref{ex-dual}. The second negation is due to Example~\ref{ex-EFL-does-not-imply-efowc}:

\begin{table}[]
     \centering
     \begin{tabular}{c|c|c|c|c|c}
          Goods & H & a & b & c & d\\
          \hline 
          Copies  & 1 & 2 & 1 & 1 & 1 \\
          \hline
           $v$ & 1000 & 100 & 1 & 2 & 2 \\
     \end{tabular}
     \caption{Instance with an EFL allocation that is not $\efowc$.}
     \label{tab:efl-not-efowc}
 \end{table}

\begin{example} ($\text{EFL} \centernot\implies \efowc$)
\label{ex-EFL-does-not-imply-efowc}
Consider three agents, and goods as given in Table~\ref{tab:efl-not-efowc}, where $v$ is the identical valuation of all agents. Then
$\allocs = \{ \{a,b\}, \{a,c,d\}, \{H\}\}$
is EFL, but not $\efowc$, as agent 1 $\efowc$-envies agent 2.  
\end{example}

 \begin{example}  ($\text{EF1} \centernot\implies \efowc$)
 Consider three agents, five goods $a,b,c,d,e$ with two copies each, and identical values $v(a) = 1$, $v(b) = v(c) = \frac{1}{2}, v(d) = v(e) = \epsilon = \frac{1}{100}$. Consider the allocation 
 $\allocs = \{ \{a, d, e\}, \{a, b, c\}, \{b, c, d, e\}$. It is EF1, and players' values are $(1.02,2,1.02)$. But agent 1 $\efowc$ envies agent 2 (the special good for the EF1 condition is $a$). For comparison,  $\allocs' = \{ \{a, b, e\}, \{a, c, d\}, \{b, c, d, e\}$
 %where players' values are $(1.51,1.51,1.02)$,
 is both EF1 and $\efowc$. 
 \end{example}
 
\subsection{MMS Approximations}
\label{sub:MMS}

%We next show that this hierarchical structure implies (strictly) different approximation guarantees to MMS.
%
%\subsection {Separation based on $\alpha$-MMS}
%
%\textcolor{red}{Remove the definition, and make the explanations later more precise to compensate (Example 6)}
%\begin{definition}
%An $\alpha$-MMS guarantee for an envy-freeness notion $\eta$ is some $\alpha$ so that for any allocation satisfying $\eta$, the allocation is $\alpha$-MMS. 
%\end{definition}
%
We show upper bounds on the MMS approximation guarantees of the three ``without commons'' notions introduced above.%
\footnote{Since our duality theorems do not hold for $\alpha$-MMS, the results in this section cannot be directly transformed to chores.} 
We also establish a lower bound for the approximation guarantee of $\eflwc$, thus separating $\efxwc$ and $\eflwc$ from $\efowc$ by a
factor that grows to infinity with $n$.
%We do not prove a separation for $\efxwc$ and $\eflwc$, their guarantees' upper and lower bound respectively are within a constant ratio of $1.2$.
It is interesting to note that the bounds we show for goods with copies are strictly lower than known bounds for goods without copies: for $\efxwc$ we give an upper bound of $0.4$ while without copies a lower bound of $\frac{4}{7}$ is known \cite{comparingNotions2018}; for $\eflwc$ we give an upper bound of $\frac{1}{3}$ while without copies a lower bound of $\frac{1}{2}$ is known \cite{barman2018groupwise}.

%\vspace{1mm}

%\noindent
%{\bf 
\paragraph{$\efowc$.} \citet{comparingNotions2018} show that EF1 %$\efowc$ 
allocations do not guarantee an approximation strictly larger than $\frac{1}{n}$ to MMS for goods without copies. Upper bounds on $\alpha$ for goods without copies immediately apply to goods with copies, since the former is a special case of the latter. 
\unless\ifdefined\full{} 
In the full version, for completeness,
\fi
\ifdefined\full{} 
For completeness,
\fi
we include an example showing that
$\efowc$ cannot guarantee strictly more than $\frac{1}{n}$-MMS for goods with copies.

%for $\efowc$. %for completeness. %in the full version. %\fi

\ifdefined\full{} 
\begin{example}[$\efowc$ cannot guarantee strictly more than $\frac{1}{n}$-MMS for goods with copies.]
\label{ex:EF1-low-MMS}
Consider an example with $n$ agents, all goods have one copy each, and identical valuations. Let $L_1,..., L_n$ be goods, all with value 1, and let $H_1,...,H_{n-1}$ be additional goods, all with value $n$. Then 
$$\allocs = \{\underbrace{\{H_{\theta},L_{\theta}\}}_{1\leq \theta \leq n-1}, \{L_n\}\}.$$
is an $\efowc$ allocation with $v(\alloci[n]) = v(L_n) = 1$, while 
 an MMS of $n$ is guaranteed for agent $n$ by 
$$\allocs ' = \{\underbrace{\{H_{\theta}\}}_{1\leq \theta \leq n-1}, \{L_1,...,L_n\}\}. $$

\end{example}
\fi

%\vspace{3mm}
%\noindent
\paragraph{\bf $\efxwc$.} We give an upper bound of $0.4$ and a lower bound of $\frac{4}{11} \approx 0.36$.

\begin{example}[There is an $\efxwc$ allocation which is at most $0.4$-MMS for goods with copies]
 
Consider 13 agents and 9 goods as given in Table~\ref{tab:efx-0.4-mms}, where $v$ is the valuation of agent 13, and all other agents value all goods as $1$. 
 \begin{table}[]
     \centering
     \begin{tabular}{c|c|c|c|c|c|c|c|c|c}
          Goods & H & x & x' & x'' & y & y' & $y_a$ & $y_b$ & $y_c$\\
          \hline 
          Copies  & 6 & 7 & 3 & 3 & 3 & 3 & 1 & 1 & 1 \\
          \hline
           $v$ & 2.5 & 1 & 1 & 1 & $\frac{1}{2}$ & $\frac{1}{2}$ & $\frac{1}{2}$ & $\frac{1}{2}$ & $\frac{1}{2}$
     \end{tabular}
     \caption{$\efxwc$ 0.4-MMS upper bound.}
     \label{tab:efx-0.4-mms}
 \end{table}
 %
 %$x$ with 7 copies, $H$ with 6 copies, $x',x'',y,y'$ with $3$ copies each, and items $y_a,y_b,y_c$ each with a single copy each. All agents but one value all goods as the same unit 1, and one unique agent has valuations $v(H) = 2.5, v(x) = v(x') = v(x'') = 1, v(y) = v(y') = v(y_a) = v(y_b) = v(y_c) = \frac{1}{2}$. 
 The following allocation is $\efxwc$:
 $$\allocs = \{\underbrace{\{H,x\}}_{\times 6}, \underbrace{\{x',x''\}}_{\times 3}, \underbrace{\{y, y', y_{\theta}\}}_{\theta \in (a,b,c)}, \{x\}\}.$$
 %
  %$$\allocs = \{\{H,x\} \times 6, \{x',x''\} \times 3, \underbrace{\{y, y', y_{\theta}\}}_{\theta \in (a,b,c)}, \{x\}\}$$
 %
 In the following allocation all agents have a value of exactly 2.5 in terms of $v$:
  $$\allocs ' = \{\underbrace{\{H\}}_{\times 6}, \underbrace{\{x,x',y\}}_{\times 3}, \underbrace{\{x, x'', y'\}}_{\times 3}, \{x, y_a, y_b, y_c\}\}.$$

 \end{example}

We establish the following $\frac{4}{11}$-MMS lower bound by appropriately generalizing a proof of \cite{comparingNotions2018} to the case of goods with copies. We intentionally keep similar notations to allow easy comparison.

\begin{theorem}
\label{thm:efx-wc_separation}
An $\efxwc$ allocation is at least $\frac{4}{11}$-MMS for goods with copies. 
\end{theorem}
\ifdefined\full{}
\begin{proof}
Suppose that allocation $\allocs$ is an $\efxwc$ allocation. Let us take the perspective of agent $\alpha$. We divide the agents into three disjoint sets: $L_1=\{i\in\agents\mid |\alloci\setminus\alloci[\alpha]|\le 1 \}$, $L_2=\{i\in\agents\mid|\alloci\setminus\alloci[\alpha]|=2 \}$, $L_3=\{i\in\agents\mid |\alloci\setminus\alloci[\alpha]|\ge3 \}$.  Define the  set of goods $\forall \theta \in \{1,2,3\}, S_{\theta} = \cup_{i\in L_{\theta}} (\alloci \setminus \alloci[\alpha]).$ 
\begin{claim}\label{claim-efxwc}
$\efxwc$ implies:
\begin{enumerate}
    \item For any good $g\in S_2$, $\vai[\alpha]{g}\le \vai[\alpha]{\alloci[\alpha]}.$
    \item For any agent $i\in L_3$, $\vai[\alpha]{\alloci\setminus\alloci[\alpha]}\le \frac{3}{2}\cdot \vai[\alpha]{\alloci[\alpha]}.$
\end{enumerate}

\end{claim}

%\unless\ifdefined\full{}
%\noindent
%The proof of Claim~\ref{claim-efxwc} is given in the supplementary version. 
%\fi
\ifdefined\full
\begin{proof}
  For the first inequality,  let $i$ be an agent in the set $L_2$. By the condition of $\efxwc$, $$\forall g\in \alloci\setminus\alloci[\alpha],\vai[\alpha]{ (\alloci\setminus\alloci[\alpha])\setminus g}\le\vai[\alpha]{\alloci[\alpha]\setminus\alloci}\le\vai[\alpha]{\alloci[\alpha]}.$$ As  $ |\alloci\setminus\alloci[\alpha]|=2$, $\vai[\alpha]{g} = \vai[\alpha]{ (\alloci\setminus\alloci[\alpha])\setminus g'} \le \vai[\alpha]{\alloci[\alpha]}$, where $g'$ is another good in the set $\alloci\setminus\alloci[\alpha]$.

   For the second inequality, $\min_{g\in \alloci\setminus\alloci[\alpha]}\vai[\alpha]{g}\le\frac{1}{2}\vai[\alpha]{\alloci[\alpha]},$ since otherwise we have for any good $g' \in A_i \setminus A_{\alpha},$
    \[
    \begin{split}
         v_{\alpha} & ((A_i \setminus A_{\alpha}) \setminus \{g'\}) = \sum_{g \in (A_i \setminus A_{\alpha}) \setminus \{g'\}} v_{\alpha}(g) \\
        & \geq \sum_{g \in (A_i \setminus A_{\alpha}) \setminus \{g'\}} \min_{g'' \in \alloci\setminus\alloci[\alpha]} \vai[\alpha]{g''} \\ & \geq  2 \min_{g'' \in \alloci\setminus\alloci[\alpha]} \vai[\alpha]{g''} > 2\cdot \frac{1}{2}v_{\alpha}(A_{\alpha}) = v_{\alpha}(A_{\alpha}), 
        \end{split}
        \]
\noindent
contradicting $\efxwc$. For $g = \arg\min_{g'\in \alloci\setminus\alloci[\alpha]}\vai[\alpha]{g'}$:
    
    \[
    \begin{split}
    & \vai[\alpha]{\alloci\setminus\alloci[\alpha]} = \vai[\alpha]{(\alloci\setminus\alloci[\alpha]) \setminus \{g\}} + \vai[\alpha]{g} \le \\
    & \vai[\alpha]{\alloci[\alpha]}+\min_{g\in \alloci\setminus\alloci[\alpha]}\vai[\alpha]{g}\le \frac{3}{2}\cdot \vai[\alpha]{\alloci[\alpha]}.
    \end{split}
    \]
\end{proof}
\fi
  
Suppose that allocation $\allocs^*$ is a maximin share allocation for agent $\alpha$. Let the allocation $\allocs'=\{\alloci\in\allocs^*\mid |\alloci\cap S_1|=0 \text{ and }  |\alloci\cap S_2|\le 1 \}$. We remove any bundle contains at least one good in $S_1$ or at least two goods in $S_2$ from the allocation $\allocs^*$.  Notice that allocation $\allocs'$ cannot be empty. As $\alpha\in L_1$, the number of bundles is at least $1+|S_1|+\frac{|S_2|}{2}$. And we remove at most $|S_1|+\frac{|S_2|}{2}$ bundles.  Let $n'=|\allocs'|$.

Next we prove that there is a bundle $\alloci[j]'$ in the allocation $\allocs'$ such that $\vai[\alpha]{\alloci[j]'}\le\frac{11}{4}\cdot \vai[\alpha]{\alloci[\alpha]}.$ Let $n'$ be the size of $|\allocs'|$, $y$ be the size of $|L_3|$ and $x$ be the number of goods in set $S_2$ appearing in the allocation $\allocs'$. 
\begin{claim}\label{claim-quan}
We have the following quantity relationships: (1) $n' \ge x$, (2) $n' \ge \frac{x}{2}+y$.
\end{claim}
%\unless\ifdefined\full{}
%\noindent
%The proof of Claim~\ref{claim-quan} is given in the supplementary version. 
%\fi
\ifdefined\full
\begin{proof}
The first inequality holds as each bundle in $\allocs'$ has at most one good from $S_2$. For the second inequality, we have $$n=|L_1|+|L_2|+|L_3|\ge |S_1|+\frac{|S_2|}{2}+y.$$ The number of removed bundles is at most $|S_1|+\frac{|S_2|-x}{2}$. Therefore, $n'\ge n-|S_1|-\frac{|S_2|-x}{2}\ge \frac{x}{2}+y$.
\end{proof}
\fi

By Claim \ref{claim-efxwc}, the total sum  of goods from sets $S_2$ and $S_3$ in allocation $\allocs'$ for agent $\alpha$ is upper bounded by 
\[
\begin{split}
    &x\cdot \max_{g\in S_2}\vai[\alpha]{g}+y\cdot \max_{i\in L_3}\vai[\alpha]{\alloci\setminus\alloci[\alpha]}\\
    \le&x\cdot \vai[\alpha]{\alloci[\alpha]}+y\cdot \frac{3}{2}\vai[\alpha]{\alloci[\alpha]}.
\end{split}
\]
 %The average valuation of goods from sets $S_2$ and $S_3$ in allocation $\allocs'$ for agent $\alpha$ is at most
By Claim \ref{claim-quan}, the average valuation is bounded by
 \begin{equation*}
\begin{split}
\frac{x+\frac{3}{2} y}{n'}\cdot \vai[\alpha]{\alloci[\alpha]} & \le \frac{x+\frac{3}{2}\cdot (n'-\frac{x}{2})}{n'}\cdot \vai[\alpha]{\alloci[\alpha]} \\
 & = (\frac{1}{4}\cdot\frac{x}{n'}+\frac{3}{2})\cdot \vai[\alpha]{\alloci[\alpha]}
% &\le \frac{\frac{1}{4}x+\frac{3}{2} x}{x}\cdot \vai[\alpha]{\alloci[\alpha]}\\
 \le\frac{7}{4}\cdot \vai[\alpha]{\alloci[\alpha]}.
\end{split}
\end{equation*}
Therefore, there is an agent $j$ such that $\vai[\alpha]{\alloci[j]'\setminus\alloci[\alpha]}\le \frac{7}{4}\cdot \vai[\alpha]{\alloci[\alpha]}.$  We have  \begin{equation*}
\begin{split}
\vai[\alpha]{\alloci[j]'}&=\vai[\alpha]{\alloci[j]'\setminus\alloci[\alpha]}+\vai[\alpha]{\alloci[j]'\cap\alloci[\alpha]}\\
&\le\frac{7}{4}\cdot \vai[\alpha]{\alloci[\alpha]}+\vai[\alpha]{\alloci[\alpha]}
=\frac{11}{4}\cdot \vai[\alpha]{\alloci[\alpha]}.
\end{split}
\end{equation*}

\noindent
MMS is $\min_{i\in\agents}\vai[\alpha]{\alloci^*}\le\vai[\alpha]{\alloci[j]'}\le\frac{11}{4}\cdot \vai[\alpha]{\alloci[\alpha]}$.
%Suppose that
\end{proof}
\fi

%\vspace{3mm}

%\noindent
\paragraph{\bf $\eflwc$.} We give an upper and a lower bound of one-third.

\begin{example}[There is an $\eflwc$ allocation with at most $\frac{1}{3}$-MMS for goods with copies]
\label{ex:third-MMS}
\unless\ifdefined\full{}
See full version for details.
\fi
\ifdefined\full{}
Consider $2\ell+1$ agents and goods as given in Table~\ref{tab:efl-1/3-mms}, where $v$ is the valuation  of agent $2\ell+1$ and all other agents value all goods as $1$. 
  \begin{table}[]
     \centering
     \begin{tabular}{c|c|c|c|c|c}
          Goods & H & x & x' & $\forall_{1\leq i \leq \ell}, \mathbf{y_i}$ & $\forall_{1\leq i \leq \ell},\mathbf{z_i}$ \\
          \hline 
          Copies  & $\ell$ & $\ell+1$ & $\ell$ & 1 & 1 \\
          \hline
           $v$ & 3 & 1 & 1 & $1 - \frac{2}{\ell}$ & $\frac{2}{\ell}$ 
     \end{tabular}
     \caption{$\eflwc$ $\frac{1}{3}$-MMS upper bound}
     \label{tab:efl-1/3-mms}
 \end{table}
%
 %and items $x$ with $k+1$ copies, items $x', H$ have $k$ copies and all the $y_i,z_i$ have 1 copy. $2k$ agents value all items as the same unit 1, and one unique agent has valuations $v(x) = v(x') = 1, v(y_i) = 1 - \frac{2}{k}, v(z_i) = \frac{2}{k}, v(H) = 3$.
Then,
 $$\allocs = \{ \underbrace{\{H,x\}}_{\times \ell}, \underbrace{\{x',y_{\theta}, z_{\theta}\}}_{1\leq \theta \leq \ell}, \{x\} \} $$ is $\eflwc$, but $v(\alloci[2\ell+1]) = 1$, while a MMS of at least $3 - \frac{2}{\ell}$ is guaranteed by $$\allocs ' = \{ \underbrace{\{H\}}_{\times \ell}, \underbrace{\{x, x',y_{\theta}\}}_{1\leq \theta \leq \ell}, \{x, z_1,...,z_{\ell}\} \}.$$
 \fi
 \end{example}
 
 %\yotam{I found there is a superadditive transition in the proof, as I marked. Also found an extension of Example 8 that gives arbitrary ratio when considering subadditive+monotone valuations. So the additive assumption is required. }
  \begin{theorem}
  \label{thm:efl-wc_separation} An $\eflwc$ allocation is at least $\frac{1}{3}$-MMS for goods with copies.
 \end{theorem}
 \ifdefined\full
 \begin{proof}
 %Consider a setting with $n$ agents, and 
 Suppose that $\allocs$ is an $\eflwc$ allocation, and consider the perspective of an agent $i^*$. 
 %Fix some $\eflwc$ allocation $\allocs$ and agent $i^*$.
 Among the remaining $n-1$ agents, let $L_1$ be the set of agents satisfying the first $\eflwc$ condition, and let $L_2$ be the remaining agents (that must thus satisfy the second condition). Denote $\ell = |L_1|, n-\ell - 1 = |L_2|$. Define $\forall \theta \in \{1,2\}, S_{\theta} = \cup_{j\in L_{\theta}} (A_j \setminus A_{i^*})$ (we allow multiple copies in the same set). Notice that for any $j\in L_2$ we have by the second $\eflwc$ condition
 that there is such good $g \in \alloci[j]$ with
 \begin{equation} \label{eq:eflwc_condition_L2}  \max \{v_{i^*}(\alloci[j] \setminus (\alloci[i^*] \cup \{g\})), v_{i^*}(g)\} \le v_i(\alloci[i^*] \setminus \alloci[j]),\end{equation}
 and so \[
 \begin{split}
      v_{i^*}(\alloci[j] \setminus \alloci[i^*]) %&\stackrel{\text{subadditive}}{\leq} v_{i^*}(\alloci[j] \setminus (\alloci[i^*] \cup \{g\})) + v_{i^*}(g) \\
      & \leq v_{i^*}(\alloci[j] \setminus (\alloci[i^*] \cup \{g\})) + v_{i^*}(g) \\
      &\stackrel{\text{Eq.~\ref{eq:eflwc_condition_L2}}}{\leq} 
      2v_{i^*}(\alloci[i^*] \setminus \alloci[j]) \\
      %&\stackrel{\text{monotone}}{\leq} 2v_{i^*}(\alloci[i^*]),
      & \leq 2v_{i^*}(\alloci[i^*]),
      \end{split}
      \]
 which implies
 %\begin{equation}
 %\label{eq:DistinctGoodsBound}
 %$$v_{i^*}(S_2) \stackrel{\text{subadditive}}{\leq} \sum_{j \in L_2} v_{i^*}(\alloci[j] \setminus \alloci[i^*]) \leq 2(n-\ell - 1) v_{i^*}(\alloci[i^*]).$$%\end{equation}
 $$v_{i^*}(S_2) \leq \sum_{j \in L_2} v_{i^*}(\alloci[j] \setminus \alloci[i^*]) \leq 2(n-\ell - 1) v_{i^*}(\alloci[i^*]).$$
 
 All goods that are not in $S_1, S_2$ must have a copy in $\alloci[i^*]$. Denote all the remaining goods' copies $R$, then we overall have $\items = S_1 \cup S_2 \cup R$. 

 %To satisfy the $\eflwc$ condition, it must hold that $k$ agents have exactly one good copy on top of shared goods with $\alloci[i]$ (we call them 'the single goods'), and for any agent $j$ of the other $(n-k)$ agents, there is such good $c$ that $$v_i(\alloci[j] \setminus (\alloci[i] \cup \{c\})), v_i(c) < v_i(\alloci[i] \setminus \alloci[j]),$$
 %and so $$v_i(\alloci[j] \setminus \alloci[i]) = v_i(\alloci[j] \setminus (\alloci[i] \cup \{c\})) + v_i(c) < 2v_i(\alloci[i] \setminus \alloci[j]) < 2v_i(\alloci[i]).$$
 %We call the goods elonging to the $(n-k)$ agents that are not shared with $\alloci[i]$ 'the distinct goods'. 
 In any allocation $\allocs ' $, there are at most $\ell$ agents with goods from $S_1$. That is since the first $\eflwc$ condition for an agent $j$ requires $|\alloci[j] \setminus \alloci[i^*]| = 1$, and so we have $|S_1| = \sum_{j\in L_1}|\alloci[j] \setminus \alloci[i^*]| = \ell$. Let $L'$ then be the set of at least $n-\ell$ agents with no good from $S_1$, and let $\{S_2^j\}_{j\in L'}, \{R^j\}_{j\in L'}$ be the allocations of $S_2, R$ goods to these agents under $\allocs '$. Since
 %$$ \sum_{j \in L'} v_{i^*}(S_2^j) \stackrel{\text{superadditive}}{\leq} v_{i^*}(S_2), $$
 $$ \sum_{j \in L'} v_{i^*}(S_2^j) \leq v_{i^*}(S_2), $$
 there exists some $j\in L'$ with 
 
 \[
 \begin{split}
  v_{i^*}(S_2^j) &\leq \frac{2(n-\ell - 1)}{|L'|}v_{i^*}(\alloci[i^*]) \\
 &\leq \frac{2(n-\ell - 1)}{n-\ell}v_{i^*}(\alloci[i^*]) < 2v_{i^*}(\alloci[i^*]).
 \end{split}
 \]
 
Since $\allocs '$ is an exclusive allocation, and since $R$ includes only goods with a copy in $\alloci[i^*]$, agent $j$ satisfies $v_{i^*}(R^j) \leq v_{i^*}(\alloci[i^*])$, and overall $v_{i^*}(\alloci[j]) = v_{i^*}(S_1^j) + v_{i^*}(S_2^j) + v_{i^*}(R^j) \leq 3v_{i^*}(\alloci[i^*])$. Since such an agent exists for any exclusive allocation, it exists for the MMS allocation, and so the minimal bundle in terms of $v_{i^*}$ in that allocation is bounded by $3v_{i^*}(\alloci[i^*])$. This shows the $\frac{1}{3}$-MMS guarantee. 
 \end{proof}
\fi

\section{Existence of $\efxwc$}
%\section{Discussion}
\label{sec:leveled}
%\textcolor{red}{A short intro for the section.}
%General existence results are not easy to attain for $\efxwc$, but are possible for restricted valuation classes.
In this section we show an existence result for leveled preferences \cite{Babaioff2017CompetitiveEW,ManjunathW21}, and identify where new ideas will be necessary to go beyond such preferences. 

\begin{definition}
A valuation $v$ is a \emph{\bf leveled preference} for goods if for any two bundles, $|B_1| > |B_2| \implies v(B_1) > v(B_2)$.
\end{definition}

We prove our existence result for goods with copies. By our duality framework (Theorem~\ref{thm:wc-duality}), %in Section \ref{sec-dual}), 
this existence result
%and algorithm
holds for chores and in fact for mixed goods and chores.
%\inbal{Is it worthwhile to give formal corollaries?}

\ifdefined\resolvedIssues

\subsection{$\efowc$}
Since we fall within the settings of \cite{cardinalityConstraints2018}, their EF1 existence result translates to our settings. Still, as explained in the previous sections, the correct notion for our settings is $\efowc$, which is a stronger notion than EF1. That is, $\efowc \implies \text{EF1}$, but not vice versa.  Intuitively, that is because the special good that must satisfy the EF1 condition for some two agents $i,j$ has to belong to the set $\alloci[j] \setminus \alloci[i]$, and can not just be in the larger $\alloci[j]$. It is thus possible for an allocation to be EF1 and not $\efowc$:

Also note that proving $\efowc$ existence for additive goods with copies gives EF1 existence for chores as a corollary, while EF1 existence for goods with copies does not.

\ifdefined\full
We now present the envy-cycle cancellation algorithm for the existence, and prove its correctness. The \emph{envy graph} of an allocation $\allocs$, denoted by $E_{\allocs}$, is the directed graph where the nodes are the set of agents, and there is an outgoing edge $i\rightarrow j$ if $\vai{\alloci} < \vai{\alloci[j]}$.   

\begin{algorithm}[H]
\SetAlgoLined
\DontPrintSemicolon
\KwIn{$\agents$, $\items$ and valuations}
\KwOut{An $\efowc$ allocation $\allocs$}
 Initialize an empty allocation $\allocs$\;
 
 \For{j=1 to $\tau$}{
    Perform a topological sort over $E_{\allocs}$, and save the first $k_j$ agents into $K$\;
    
    Add a copy of the good $j$ for each of the agents in $K$\;
    
    Cancel cycles in $E_{\allocs}$\;
 }
 \caption{$\efowc$ allocation for a general additive goods with copies instance}
 \label{alg:EF1-WC}
\end{algorithm}

\begin{theorem}
Algorithm~\ref{alg:EF1-WC} always runs to completion and produces an $\efowc$ allocation. 
\end{theorem}
\begin{proof}
We prove two properties of $\allocs$ that are always maintained at the beginning (end) of the loop round:

\emph{$E_{\allocs}$ is acyclic} - for the empty allocation $\allocs$, $E_{\allocs}$ is an empty graph. Step 5 guarantees that $E_{\allocs}$ remains acyclic with each round of the loop, by using the primitive 

\emph{$\allocs$ has no $\efowc$-envious agents} - This holds for the empty allocation. Each round of the loop maintains that if a copy of a good was added to a bundle an agent envies, then it was also added to its own bundle (by the topological sorting). If a copy of a good was added to a bundle the agent does not envy, then removing it removes envy, which suffices for the $\efowc$ condition. The cycle cancellation preserves envy-freeness notions, as mentioned.  
\end{proof}
\fi

\unless\ifdefined\full
\begin{theorem}
There is a polynomial time algorithm that finds an $\efowc$ allocation for any additive goods with copies instance. 
\end{theorem}

\fi

\subsection{$\efxwc$ for (Almost) Identical Ordinal Preferences}
An instance is identical ordinal if all agents have the same ordering for all goods. \cite{IdenticalOrdinalPref} designs two efficient algorithms for goods and chores separately to find an EFX allocation in standard setting for identical ordinal preferences. In this section, we extend their result to the notion $\efxwc$ and a more general valuation class. Due to the duality, one algorithm for goods is sufficient for our model. 

%We give an algorithm that finds an $\efxwc$ allocation for this setting. Since it works for any number of copies for each good, it specifically solves the case where all items have $n-1$ copies, which by  Theorem~\ref{thm:wc-duality} is equivalent to solving the EFX problem for chores (under similar restrictions, where all agents up to one have the same ordinal preferences). 

{\bf Almost IO instance:}  There is an special agent $\omega$ and an ordering $\succ_{id}$ of goods $\types$ such that $g\succ_{id}g'$, we have $\vai{g}\ge\vai{g'}$ for all agents $i\neq\omega$.
%An instance such that at least $n-1$ of the agents' valuation functions are in line with  some order $>_{id}$ over single goods. We denote the unique agent (or an arbitrary agent, if all ordinal preferences are identical) $unq$. 

We make use of the primitive of the envy graph cycle cancellation introduced by \cite{lipton2004approximately}. %Given some allocation and a cycle of envies in $E_{\allocs}$, we give each agent the bundle it envies. This maintains any envy-free notion, since the set of bundles stays the same, while each agent is assigned a better bundle by her valuation.
%{\bf Envy graph:} 
The \emph{envy graph} of an allocation $\allocs$, denoted by $E_{\allocs}$, is the directed graph where the nodes are the set of agents, and there is an outgoing edge $i\rightarrow j$ if $\vai{\alloci} < \vai{\alloci[j]}$.

% \textcolor{Try to avoid new notation and stick to words}

%For $n$ agents and $k_1,...,k_t$ copies of goods, assume there are $n-1$ agents with the same ordinal preference over the goods, and one agent with some arbitrary ordinal preference over the goods. 

%\xin{It is not necessary to limit each goods must have $k$ copies. The algorithm can work for arbitrary copies. This is similar to our definition of exclusive allocation without a general number for copies.}

\begin{algorithm}[H]
\SetAlgoLined
\DontPrintSemicolon
\KwIn{
An almost IO instance
%$n$ agents with $n-1$ identical ordinal preferences, $m$ indivisible goods with $k_1,...,k_t$ copies each, $\{v_i^{id}\}_{i\in \mathcal{N} \setminus \{unq\}}, v^{unq}$ valuation functions
}
\KwOut{$\allocs$, an $\efxwc$ allocation}
 Order the goods $g_1,\cdots,g_{\tau}$ by decreasing valuation according to the ordering $\succ_{id}$\;
 
 Initialize an empty allocation $\allocs$\;
 
 \For{j=1 to \tau}{
    Let $G=E_{\allocs}$ ($G$ does not change when $\allocs$ changes)\;
    
    Perform a topological sort over $G$ with priority to the agents who have a directed path going to agent $\omega$\;
    
    Let $K$ be the first $k_j$ agents according to order of topological sort\;
    
  %  Create a temporary allocation $\allocs' = \allocs$ 
    
    Add a copy of the good $g_j$ for each of the agents in $K$\;\label{inalg-add}
    
   \If{ the best bundle for agent $\omega$ is  $h\in K$} {
        %Let $h_{\max}$ be the agent in $K$ that has the highest valuation according to $unq$
        
        Let $s_0 = h,...,s_r = \omega$ be a path in $G$ from $h$ to $\omega$ such that if $s_i$ has a successor in $K$, then $s_{i+1}\in K$\; %chosen as in Remark~\ref{rem:pathChoice}\;
        
        Perform the cyclic reassignment $\allocs_{s_i} = \allocs_{s_{(i+1) mod (r + 1)}} $\;
    }
    Cancel cycles in $E_{\allocs}$\;
 }
 \caption{$\efxwc$ allocation for almost identical ordinal preference}
 \label{alg:EFX-WC_ordinal}
\end{algorithm}

%\begin{remark}
%\label{rem:pathChoice}
%Let us describe the choice of the path from $h_{\max}$ to $unq$. Such a path always exists. For $h_{\max}$, there must be a node appearing in the topological sort between itself and $unq$ (including $unq$) that it envies, otherwise, $h_{\max}$ could appear after $unq$ in the topological sort, contradicting our preference for $unq$ to appear early. This property carries recursively to any node chosen by $h_{\max}$, and so on. When choosing the path recursively, at each step we may prioritize choosing the next node to be in $K$, and choose a node outside $K$ only if the agent does not envy any node in $K$. Notice that once we choose an agent outside $K$, all subsequent agents are outside $K$, by the topological sort. 
%\end{remark}

\begin{theorem}
Algorithm~\ref{alg:EFX-WC_ordinal} always outputs an $\efxwc$ allocation in $O(\tau n^2)$ time. 
\end{theorem}

\ifdefined\full
\begin{proof}
We prove that the algorithm always maintains the allocation $\allocs$ an $\efxwc$ allocation.

We first prove until line 

\end{proof}
\fi
\fi

%\subsection{$\efxwc$ for Leveled Preferences}

\begin{theorem}
%Given any item set $\items$ and agent set $\agents$, 
There is an algorithm that always finds an $\efxwc$ exclusive allocation for goods with copies in the case of leveled preferences. Its runtime is $O(n|\types|^2)$.
%of each agent is a leveled preference. % (could be non-additive).
\end{theorem}

\unless\ifdefined\full{}
\begin{proof}
We describe here the algorithm and defer the details to the full version. We choose an initial allocation $\allocs$ such that $|\alloci|-|\alloci[j]|\le 1$ for every pair of agents $i,j\in\agents$, e.g., by setting some arbitrary order $1,...,n$ over the agents, an arbitrary order $1,...,t$ over the good types, and allocating the $k_i$ copies of the next good to the next $k_i$ agents in a cyclic fashion. 
If all bundle sizes in $\allocs$ are the same, then $\allocs$ is $\efxwc$ (and moreover, EFX), we are done.

Otherwise, $\allocs$ has two bundle sizes (two levels). Let the number of goods in the upper level of $\allocs$ be $H$, and thus the number of goods in the lower level is $H-1$. If $\allocs$ is not $\efxwc$ we perform the following operation. Suppose that agent $i$ $\efxwc$-envies agent $j$. %(that is, EFX-envies it after removing the goods they have in common)
Agent $i$ must be at a lower level than agent $j$, otherwise after removing a good from $j$, agent $i$ is at a higher level and by the leveled preferences prefers its own bundle. 
Let $g_{\max} = \arg\max_{g \in \alloci[j] \setminus \alloci[i]} v_i(g), g_{\min} = \arg\min_{g \in \alloci[i] \setminus \alloci[j]} v_i(g)$.
%Let $g_{\max} = arg\max_{g \in \alloci[j]} v_i(g), g_{\min} = arg\min_{g \in \alloci[i]} v_i(g)$. 
Let agent $i$ get the bundle $(\alloci[i]\setminus \{g_{\min}\}) \cup \{g_{\max}\}$, and let agent~$j$ get the bundle $(\alloci[j]\setminus \{g_{\max}\}) \cup \{g_{\min}\}$. After this operation, the allocation remains exclusive. Agent $i$ gets a strictly improved bundle by her valuation (as we can prove), and both agents get a bundle with the same cardinality as before, thus maintaining the sets of lower-level and upper-level bundle agents unchanged. We repeat the process. 
\end{proof}
\fi

\ifdefined\full{}

\begin{proof}
We can choose an initial allocation $\allocs$ such that $|\alloci|-|\alloci[j]|\le 1$ for any $i,j\in\agents$, e.g., by setting some arbitrary order $1,...,n$ over the agents, an arbitrary order $1,...,t$ over the good types, and allocating the $k_i$ copies of the next good to the next $k_i$ agents in a cyclic fashion. If all bundle sizes in $\allocs$ are the same, then $\allocs$ is $\efxwc$ (and moreover, EFX), we are done.

Otherwise, $\allocs$ has two bundle sizes (two levels). Let the number of goods in the upper level of $\allocs$ be $H$, and thus the number of goods in the lower level is $H-1$. If the allocation $\allocs$ is not an $\efxwc$ allocation, we perform the following operation. Suppose that agent $i$ $\efxwc$-envies agent $j$ (that is, EFX-envies it after removing the goods they have in common). Agent $i$ must be at a lower level than agent $j$, otherwise after removing a good from $j$, agent $i$ is at a higher level and by the leveled preferences prefers its own bundle. 

It must hold that $\max_{g \in \alloci[j] \setminus \alloci[i]} v_i(g) > \min_{g \in \alloci[i] \setminus \alloci[j]} v_i(g)$, otherwise we have for any good $g' \in \alloci[j] \setminus \alloci[i]$:
\[
\begin{split}
    & v_i((\alloci[j] \setminus \alloci[i]) \setminus \{g'\}) \leq |(\alloci[j] \setminus \alloci[i]) \setminus \{g'\}|\max_{g \in \alloci[j] \setminus \alloci[i]} v_i(g) \leq \\
    & |\alloci[i] \setminus \alloci[j]|\min_{g \in \alloci[i] \setminus \alloci[j]} v_i(g) \leq v_i(\alloci[i] \setminus \alloci[j]), 
\end{split}
\]

in contradiction to our assumption of $\efxwc$-envy. 

Let $$g_{\max} = \arg\max_{g \in \alloci[j] \setminus \alloci} v_i(g), g_{\min} = \arg\min_{g \in \alloci[i] \setminus \alloci[j]} v_i(g).$$ Let agent $i$ get the bundle $(\alloci[i]\setminus \{g_{\min}\}) \cup \{g_{\max}\}$, and let agent $j$ get the bundle $(\alloci[j]\setminus \{g_{\max}\}) \cup \{g_{\min}\}$. After this operation, the allocation remains exclusive. Agent $i$ gets a strictly improved bundle by its valuation, and both agents get a bundle with the same cardinality as before, thus maintaining the sets of lower-level and upper-level bundle agents unchanged. We repeat the process.

We construct the potential function to show the number of steps is bounded and polynomial. For any good $g\in T$, let $\omega_i(g)$ be its ordinal position according to agent $i$'s preference over the goods, e.g., for the minimal good $g\in T$ by $i$ valuation we have $\omega_i(g) = 1$, and for the maximal good $g'$ we have $\omega_i(g') = |\types|$. We consider the potential function $$\psi(\allocs) = \sum_{\substack{i\in \mathcal{N} \\ |\alloci[i]| = H-1}} \sum_{g\in \alloci[i]} \omega_i(g).$$
Notice that at each step this potential function strictly increases as we replace some good with a strictly preferred good for some lower level agent. Also note that $0 \leq \min_{\allocs \in \ea{\items}} \psi(\allocs) \leq \max_{\allocs \in \ea{\items}} \psi(\allocs) \leq n |\types|^2$, and the function always returns an integer value. Thus the maximal number of substitution steps is in $O(n|\types|^2)$. 
 \end{proof}
\fi

%\begin{remark} 
\noindent{\bf Challenges.}
Many existence results of envy-based fairness notions for goods without copies rely on the primitive of envy-cycle canceling, first shown by \cite{lipton2004approximately}. We thus note an important technical difference in proving existence for $f_{wc}$ notions. For goods with copies, it is sometimes impossible to cancel an envy-cycle without breaking the fairness notion. 
This was first pointed out in \cite{EnvyCycleCancellationExample}. We give below another such example that has two additional properties: First, the allocation is $\efxwc$ before cancelling the envy-cycle, but not even $\efowc$ after the cancellation. Second, in our example the choice of which envy-cycle to cancel is immaterial to the difficulty arising, as there is only one envy cycle. 
%The following example demonstrates it. 
%\end{remark}

\begin{comment}
  \begin{table}[]
     \centering
     \begin{tabular}{c|c|c|c|c|c|c|c|c}
          Goods & H_1 & H_2 & a & b & c & d & e & f\\
          \hline 
          Copies  & 1  & 2 & 1 & 2 & 2 & 1 & 2 &  1   \\
          \hline
           $\vau[1]$ & 2.1 & 2 & 1 & 1 &  1 & 0.1 & 0.1 & 0.1 \\
          \hline
          $\vau[2]$  & 4 & 2 & 1 & 0.5 &  0.5 & 0.1 & 0.1 & 0.1\\
          \hline
          $\vau[3]$ &  2.1 & 2 & 1 & 0.5 &  0.5 & 0.1 & 0.1 & 0.1\\
          \hline
          $\vau[4]$ & 4 & 1.1 & 1 & 1 &  1 & 1 & 1 & 1
     \end{tabular}
     \caption{Envy-cycle cancelation failure for goods with copies}
     \label{tab:cycle-cancle}
 \end{table}
 
 \begin{example}
 There are 4 agents and 8 good types. The valuation and the number of copies are listed in Table \ref{tab:cycle-cancle}. Note that these valuations have identical ordinal preferences. Consider the allocation $\alloci[1]=\{b,c,d,e,f\}$, $\alloci[2]=\{H_2,a\}$, $\alloci[3]=\{H_2,b,c,e\}$, $\alloci[4]=\{H_1\}$. It is an $\efxwc$ allocation. We have agent $1$ envies agent 2, agent 2 envies agent 4 and agent 4 envies agent 1. Let us reallocate the bundles to resolve this envy cycle. Even though everyone gets a better bundle, it is not $\efowc$, as agent 1 then gets the bundle $\{H_2,a\}$ but $\efowc$-envies the bundle $\{H_2,b,c,e\}$.
 \end{example}
\end{comment}
 \begin{table}[]
     \centering
     \begin{tabular}{c|c|c|c|c|c|c|c}
          Goods & H & a & b & c & d & e & f\\
          \hline 
          Copies  & 2 & 1 & 1 & 2 & 2 & 2 &  1   \\
          \hline
           $\vau[1]$ & 2.5 & 1 & 1 &  1 & 1 & 0.1 & 0.1 \\
          \hline
          $\vau[2]$  & 2 & 1.5 & 1 &  0.7 & 0.7 & 0.7 & 0.7\\
          \hline
          $\vau[3]$ &  2 & 1 & 0.5 &  0.5 & 0.1 & 0.1 & 0.1\\
     \end{tabular}
     \caption{Envy-cycle cancellation failure for goods with copies.}
     \label{tab:cycle-cancle-small}
 \end{table}

 \begin{example}
 There are 3 agents and 7 good types. The valuation and the number of copies are listed in Table \ref{tab:cycle-cancle-small}. Note that these valuations have identical ordinal preferences. Consider the allocation $\alloci[1]=\{b,c,d,e,f\}$, $\alloci[2]=\{H,a\}$, $\alloci[3]=\{H,c,d,e\}$. It is an $\efxwc$ allocation. We have that agent~$1$ envies agent 2 and vice versa. Let us reallocate the bundles to resolve this envy-cycle. Even though everyone gets a better bundle, it is not $\efowc$, as agent 1 then gets the bundle $\{H,a\}$ but $\efowc$-envies the bundle $\{H,c,d,e\}$.
 \end{example}
 %\begin{remark}
Another intriguing property of goods with copies is that the MNW allocation is not even necessarily $\efowc$ (unlike the case of goods), as the following example shows. 
%\end{remark}
\begin{definition}
The \emph{\bf Nash welfare} of an allocation is the product $NW(\allocs) = \prod_{i\in \agents} v_i(\alloci)$.
The \emph{\bf maximum Nash welfare} allocation is MNW $= \arg \max_{\allocs} NW(\allocs)$. 
\end{definition}
 \begin{table}[]
     \centering
     \begin{tabular}{c|c|c|c|c}
          Goods &  a & b & c & d \\
          \hline 
          Copies  & 2 & 1 & 1 & 1  \\
          \hline
           $\vau[1]$ & 1 & 1 & 1 &  $\epsilon$ \\
          \hline
          $\vau[2]$  & 1 & $\epsilon$ & $\epsilon$ &  $\epsilon$ \\
          \hline
          $\vau[3]$ &  $\epsilon$ & $\epsilon$ & $\epsilon$ &  1 \\
     \end{tabular}
     \caption{An MNW allocation that is not $\efowc$.}
     \label{tab:mnw-not-ef1wc}
 \end{table}
\begin{example}
Consider Table~\ref{tab:mnw-not-ef1wc} with $\epsilon = 10^{-6}$. The MNW allocation is the exclusive allocation $\{\{a,b,c\}, \{a\}, \{d\} \}$. In this allocation agent 2 $\efowc$-envies agent 1. 
\end{example}

%\section{Separation Results}

 \section{Discussion}
 \label{sec:discussion}
 %\textcolor{red}{Write a draft}
 To the best of our knowledge, we provide the first formal duality relationship between goods with copies and chores that establishes the equivalence of both settings for a broad class of fairness notions. Some of our results can be extended beyond additive valuations. In particular, Theorem~\ref{thm:wc-duality} holds for general valuations, and Proposition~\ref{prop:inclusive_efxwc_eflwc} holds for all monotone valuations, as indicated in the proof. Other results do not generalize. For example, Theorem~\ref{thm:efl-wc_separation} can not be extended, even to monotone and submodular valuations:
 
 %\begin{definition}
%A valuation $v$ of goods is \emph{subadditive} if $\forall S_1,S_2 \subseteq \items, v(S_1) + v(S_2) \geq v(S_1 \cup S_2)$.
%\end{definition}

 \begin{definition}
A valuation $v$ of goods is \emph{\bf submodular} if $\forall S_1,S_2 \subseteq \items, v(S_1) + v(S_2) \geq v(S_1 \cup S_2) + v(S_1 \cap S_2)$.
\end{definition}
 
%\begin{example}[For monotone and subadditive valuations, there is an $\eflwc$ allocation with at most $\frac{1}{n}$-MMS for goods with copies]
%\label{ex:third-MMS-subadditive}
%Consider $n$ agents and goods $\{y^{i}_{j}\}^{1\leq i\leq n}_{1\leq j \leq n}$, each with one copy, where $v$ is the valuation function of all agents:
%$$v(S) = |\{i\ |\ \exists j, y^i_j \in S\}|.$$
%In words, there are $n^2$ goods, partitioned into $n$ groups, where goods in the same group are substitutes, and the valuation of a bundle of goods equals to the number of good groups present in the bundle. Then, 
 %$$\allocs = \{ \underbrace{\{y^{\theta}_{1}, \ldots,  y^{\theta}_n\}}_{1\leq \theta \leq n} \} $$ is $\eflwc$ (in fact it is envy-free), but $v(\alloci[i]) = 1$ for each agent $i$, while a MMS of $n$ is guaranteed by $$\allocs ' = \{ \underbrace{\{y^{\theta}_1, \ldots, y^{(\theta + i) \mod n}_i, \ldots, y^{(\theta + n-1) \mod n}_n\}}_{1\leq \theta \leq n} \}. $$
 %\end{example}
 
 \begin{example}[For monotone and submodular valuations, there is an $\eflwc$ allocation with at most $\frac{1}{n}$-MMS for goods with copies]
\label{ex:third-MMS-submodular}
\unless\ifdefined\full{}
See full version for details.
\fi
\ifdefined\full{}
Consider $n$ agents and goods $\{y^{i}_{j}\}^{1\leq i\leq n}_{1\leq j \leq n}$, each with one copy, where $v$ is the valuation function of all agents:
$$v(S) = |\{i\ |\ \exists j, y^i_j \in S\}|.$$
In words, there are $n^2$ goods, partitioned into $n$ groups, where goods in the same group are substitutes, and the valuation of a bundle of goods equals to the number of good groups present in the bundle. Then, 
 $$\allocs = \{ \underbrace{\{y^{\theta}_{1}, \ldots,  y^{\theta}_n\}}_{1\leq \theta \leq n} \} $$ is $\eflwc$ (in fact it is envy-free), but $v(\alloci[i]) = 1$ for each agent $i$, while a MMS of $n$ is guaranteed by $$\allocs ' = \{ \underbrace{\{y^{\theta}_1, \ldots, y^{(\theta + i) \mod n}_i, \ldots, y^{(\theta + n-1) \mod n}_n\}}_{1\leq \theta \leq n} \}. $$
 \fi
 \end{example}
 
For future work, we believe it is of interest to further investigate the duality phenomenon studied in this paper in the context of other fair division settings, e.g., divisible items.
%It is interesting to continue to get a deeper understanding of this duality phenomena about fair division.  
The main open challenge for the new fairness notions that we propose (such as $\efxwc, \efowc$) is to settle their general existence for goods with copies. As a special case, the existence of EFX, EFL for chores is open.
 %, as well as special cases such as identical ordinal valuations or lexicographic valuations for goods with copies could be of interest.
 %
 It is also worth investigating a more flexible model where the number of copies of each good is only loosely set, e.g., constrained between a minimal and maximal value.
 %A more flexible model worth investigating is one where the number of copies of each good are more loosely set, e.g., are constrained between a minimal and maximal value. In this case existence may be easier to show. 
 
 \section*{Acknowledgments}
 Yotam Gafni and Ron Lavi were partially supported by the ISF-NSFC joint research program (grant No.~2560/17). Xin Huang was partially supported by the Aly-Kaufman Fellowship. This research was supported by the Israel Science Foundation (grant No.~336/18). 
 
 An early version of this work was presented at the 8th International Workshop on Computational Social Choice (COMSOC), 2021.% 
 \footnote{The COMSOC version is accessible at \url{https://drive.google.com/open?id=17b5r_kXT8lC1fzvKkU7-0Vh_WRTAIosk}.} %{here}. 
%TODO AFTER SUBMISSION:
 
%\subsection{Separation Based on MMS-$k$}

%Time and space allowing.

%\section{EFL existence}

\bibliography{aaai22.bib}

\ifdefined\aps{}
\appendix
\label{app:APS}

\section{Appendix: Duality for AnyPrice Share}

We now show a sense in which duality holds for the AnyPrice Share \cite{fairshare_arbitraryentitlements}, without using Theorem~\ref{thm-meta-share}. First, we define the notion of fair shares with entitlements:

\begin{definition}
\label{def:share-fair-entitlements}
Given a valuation $v$, %the number of agents $n$ and 
an item set $\items$, and an entitlement $b$, a {\bf share} function $s$ outputs a real value, i.e., $s(v,\items, b)\in \mathbb{R}$. An exclusive allocation $\allocs$ is called {\bf $s$-share fair} if we have $\vai{\alloci}\ge s(\vau[i],\items, b_i), \forall i\in\agents$. 
\end{definition}

We next define the notion of AnyPrice Share. 

\begin{definition}
(AnyPrice Share for goods and chores)
Let $\Delta(\items)$ be the set of probability distribution over items $\items$ such that the probability of the same type items are the same. 

For goods:
    $$\text{APS}(i,\val,b_i)= \min_{\mathbf{p} \in \Delta(\items)}\max_{\allocs \in \ea{\items}  }\left\{\vai{\alloci}\mid  \sum_{g\in \alloci}p_g\le b_i \right\}$$
    
%    where $sign(v) = \begin{cases} 1 & v \geq 0 \\
%-1 & v < 0
%\end{cases}$. 
    
For chores:  
     $$\text{APS}(i,\val,b_i)= \max_{\mathbf{p} \in \Delta(\items)}\min_{\allocs\in \ea{\items}  }\left\{\vai{\alloci}\mid  \sum_{c\in \alloci}p_c\ge b_i \right\}$$

For chores, this takes a new interpretation: Rather than having to choose items with prices within a budget, the agent has to take chores of at least total weight of $b$. 

\end{definition}

\begin{claim}
If we add $n$ copies of the same good $g$, APS moves by the value of the item. 
\end{claim}
\begin{proof}
Choose a price $\alpha \leq \frac{1}{n}$ for the good. We know that there is a price vector for the other items with
$$\sum p_c = 1, \sum_{c\in APS\_set} p_c \leq \frac{1}{n},$$
After setting price $\alpha$ for the extra item, we have $1 - n \alpha$ budget left for the other items. Scale the vector $p_c \rightarrow (1 - n \alpha) p_c$, then
$$\sum_{c\in APS\_set} p_c (1 - n\alpha) + \alpha \leq \frac{1}{n}(1 - n\alpha) + \alpha = \frac{1}{n},$$
and so this price vector shows that $APS$ of the new instance is at least $APS + v(g)$. 
\end{proof}

\begin{claim}
If we have goods with $k$ copies (all of them), then APS is dual. 
\end{claim}
\begin{proof}
We consider the prices for all the copies of the same type as $p_t$, they satisfy $kp_i = p_t$ for a specific item $i$ of that type. 

We have for the APS set
$$\sum_{i \in APS\_set} p_i = \sum_{i \in APS\_set}  \frac{1}{k} p_t \leq \frac{1}{n},$$
i.e.,
$$1 - \sum_{i \in T \setminus (APS\_set)} p_t = \sum_{i \in APS\_set} p_t \leq \frac{k}{n},$$
i.e.,
$$\frac{n-k}{n} \leq \sum_{i \in T \setminus (APS\_set)} p_t,$$
i.e., 
$$\frac{1}{n} \leq \sum_{i \in T \setminus (APS\_set)} \frac{1}{n-k} p_t = \sum_{i \in T \setminus (APS\_set)} p_i,$$
which is the APS condition for chores (since in the dual instance $(n-k)p_i = p_t$. 

So we conclude that the APS for chores is at least $APS - v_i(T)$, and the reverse argument will show it's exactly so. 

\end{proof}

We first demonstrate why this is the natural extension of the APS definition to chores using a chores version of Example 5 of \cite{bestofbothworlds}. 

\begin{example}
\label{ex:APS}
Consider $n=4$ agents, 5 chores with 3 copies each, valuations $v = -2,...,v_5 = -6$, and budget $b = \frac{3}{4}$. We have $PROP = -3 \cdot \frac{2 + ... + 6}{4} = -15$. APS is at most $-16$, by considering prices $p_1 = p_2 = 0.125, p_3 = 0.16, p_4 = 0.251, p_5 = 0.329$. It is at least $-16$. Assume by contradiction otherwise: Then any bundle with lower valuation that is part of an exclusive allocation must include chores $3,4$ and $5$, as well as at least one of the chores $1,2$. Thus, it must hold that $\min\{p_3,p_4,p_5\} > 0.25$ (Otherwise this chore can be excluded from the bundle). Therefore, $p_3 + p_4 + p_5 > 0.75$, which results in $APS = -15$. We conclude that $APS = -16$, which has two nice properties: It maintains $APS \leq PROP$ as in the goods case, and it is exactly $-v(\items) = 20$ lower than the $APS$ of the dual goods instance with $n=4$ agents, the reverse (positive) valuation $-v$, one copy of each good, and a budget of $b = \frac{1}{4}$. For a linear share, we would expect to have this property with regards to the dual goods instance with the \emph{same} budget of $b = \frac{3}{4}$, but this does not hold. 
\end{example}

We now see by the following proposition that the last property is indeed general and matches the results of Theorem~\ref{thm-meta-share}. 

\begin{prop}
For APS, an exclusive allocation $\allocs$ is $s$-share fair for $\items, \val, b_i$ iff $\dual{\allocs}$ is $s$-share fair for $\dual{\items}, \dual{\val}, 1-b_i$.
\end{prop}
\begin{proof}
We assume a goods instance ($v_i$ is non-negative for any agent $i \in \agents$). The argument for the converse is similar. 

As noted in the proof of Theorem~\ref{thm-meta-share}, item sets $\items$ and $\dual{\items}$ are linearly related w.r.t. valuations $\vai, \dual{\vau}$ with $d = \vai{\types}$.

Suppose that $\allocs$ is $s$-share fair for $\items, v_i, b$. By definition, we have $\vai{\alloci}\ge s(\vau,\items, b)$ for all $i\in\agents$. 

\[
\begin{split}
& s(\dual{\vau},\dual{\items}, 1 - b) =  \min_{\mathbf{p} \in \Delta^{|\types|}}\max_{\substack{\dual{\alloc} \in \mathcal{X}(\dual{\items}) \\ \sum_{j\in \dual{\alloci}}p_j \geq 1 - b}}- \vai{\dual{\alloci[j]}} = \\
& \min_{\mathbf{p} \in \Delta^{|\types|}}\max_{\substack{\dual{\alloc} \in \mathcal{X}(\dual{\items}) \\ 1 - \sum_{j\in \dual{\alloci}}p_j \leq b}}- \vai{\dual{\alloci[j]}} = \\
& \min_{\mathbf{p} \in \Delta^{|\types|}}\max_{\substack{\dual{\alloc} \in \mathcal{X}(\dual{\items}) \\ \sum_{j\in \types \setminus \dual{\alloci}}p_j \leq b}}- \vai{\dual{\alloci[j]}} = \\
& -\vai{\types} + \min_{\mathbf{p} \in \Delta^{|\types|}}\max_{\substack{\dual{\alloc} \in \mathcal{X}(\dual{\items}) \\ \sum_{j\in \types \setminus \dual{\alloci}}p_j \leq b}} \vai{\types} - \vai{\dual{\alloci[j]}} = \\
& -\vai{\types} + \min_{\mathbf{p} \in \Delta^{|\types|}}\max_{\substack{\alloc \in \mathcal{X}(\items) \\ \sum_{j\in \alloci}p_j \leq b}} \vai{\alloci[j]} = s(\vau, \items, b) -\vai{\types}
\end{split}
\]

We now have
\[
\begin{split} \dual{\vau}(\dual{\alloci})  = & \vai{\alloci} - \vai{\types} \geq s(\vau, \items, b) - \vai{\types} = \\
& s(\dual{\vau}, \dual{\items}, 1-b). 
\end{split}
\]
\end{proof}

\begin{remark}
We remark that: 

\begin{itemize} 

\item Although not linear, the entitlement transformation $b_i \rightarrow 1 - b_i$ has an intuitive dual sense of a complement, as all entitlements are normalized with respect to $1$. 

\item A useful byproduct of the duality meta-theory is the ability to substantiate a correct translation of notions from goods to chores settings. To the best of our knowledge, AnyPrice Share was not defined for chores prior to this work. By providing a definition, and proving it is dual with the goods notion, we find that this substantiates it being the correct translation for chores. 
\end{itemize}
\end{remark}

\fi

\end{document}